\documentclass[12pt]{article}
\usepackage{geometry}
\usepackage[utf8]{inputenc}

\usepackage{hyperref}

\usepackage{amssymb,amsmath,amsthm,stmaryrd,comment}

\usepackage{enumerate}

\usepackage{tikz}

\usetikzlibrary{
  calc,  
  cd,     
}

\tikzset{
vertex/.style={circle,draw,inner sep=.75pt,fill=black}
}

\theoremstyle{plain}
\newtheorem{theorem}{Theorem}
\newtheorem{lemma}[theorem]{Lemma}
\newtheorem{proposition}[theorem]{Proposition}
\newtheorem{corollary}[theorem]{Corollary}

\theoremstyle{definition} 

\newtheorem{remark}[theorem]{Remark}

\newtheorem{notation}[theorem]{Notation}
\newtheorem{open}[theorem]{Open Problem}

\newcounter{oq}

\newcounter{claim}
\renewcommand{\theclaim}{\Alph{claim}}
\newenvironment{claim}{\refstepcounter{claim}%
\par\medskip\par\noindent{\it Claim~\theclaim.~}~\rm}%
{\par\smallskip\par}
\newenvironment{subproof}{\bigskip\par\noindent{\sl Proof of Claim~\theclaim.~}}%
{$\,\triangleleft$\par\medskip\par}

\newcommand{\Case}[2]{\medskip\par{\it Case #1:\/ #2}}
\newcommand{\subcase}[2]{\smallskip\par{\it Subcase #1:\/ #2}}

\makeatletter
\def\@gifnextchar#1#2#3{\let\@tempe#1\def\@tempa{#2}\def\@tempb{#3}%
  \futurelet\@tempc\@gifnch}

\def\@gifnch{\ifx\@tempc\@sptoken\let\@tempd\@tempb%
  \else\ifx\@tempc\@tempe\let\@tempd\@tempa\else\let\@tempd\@tempb\fi\fi\@tempd}

\def\SK@set#1{\left\{#1\right\}}
\def\SK@@set#1#2{\{#1\,:\,
    \begin{array}{@{}l@{}}#2\end{array}
\}}
\def\SK@mset#1{\left\{\!\!\left\{#1\right\}\!\!\right\}}
\def\SK@@mset#1#2{\{\!\!\{#1\,:\,
    \begin{array}{@{}l@{}}#2\end{array}
\}\!\!\}}
\def\BIG@set#1{\Big\{#1\Big\}}
\def\BIG@@set#1#2{\Big\{#1\:\Big|\:
    \begin{array}{@{}l@{}}#2\end{array}
\Big\}}
\newcommand{\Set}[1]{\@gifnextchar\bgroup{\SK@@set{#1}}{\SK@set{#1}}}
\newcommand{\Mset}[1]{\@gifnextchar\bgroup{\SK@@mset{#1}}{\SK@mset{#1}}}
\newcommand{\Bigset}[1]{\@gifnextchar\bgroup{\BIG@@set{#1}}{\BIG@set{#1}}}
\makeatother

\newcommand{\refeq}[1]{(\ref{eq:#1})}

\newcommand{\of}[1]{\left( #1 \right)}
\newcommand{\function}[2]{:#1 \rightarrow #2}
\newcommand{\compl}[1]{\overline{#1}}

\newcommand{\bR}{\mathbb{R}}

\newcommand{\jj}{\mathrm{j}}

\newcommand{\kwlout}[1]{\mathit{WL}_{#1}}

\DeclareMathOperator{\ea}{\mathsf{EA}}
\DeclareMathOperator{\wm}{\mathsf{WM}}
\DeclareMathOperator{\mea}{\mathsf{main-EA}}
\DeclareMathOperator{\spec}{\mathsf{Spec}}
\DeclareMathOperator{\genspec}{\mathsf{gen-Spec}}
\DeclareMathOperator{\weak}{\mathsf{weak-FSI}}
\DeclareMathOperator{\strong}{\mathsf{strong-FSI}}
\newcommand{\wlk}[1]{\mathsf{WL}_{#1}}

\newcommand{\wlkk}[1]{\mathit{WL}_{#1}}
\newcommand{\ui}[2]{#1^{(#2)}}
\newcommand{\chir}[1]{\ui{\chi}{#1}}
\newcommand{\ppr}[1]{\ui{P_*}{#1}}
\newcommand{\wwr}[1]{\ui{w_*}{#1}}
\newcommand{\calC}{\mathcal{C}}
\newcommand{\calI}{\mathcal{I}}

\newcommand{\canononexfirst}[1]{#1(x,x)}
\newcommand{\canononexsubsecond}[1]{( #1(x,y), #1(y,y) ) } 
\newcommand{\canononexsecond}[1]{\Mset{ \canononexsubsecond{#1} }_y }
\newcommand{\canononexsecondG}[1]{\Mset{ \canononexsubsecond{#1} }_{y\in V(G)} }
\newcommand{\canononex}[1]{\of{ \canononexfirst{#1}, \canononexsecond{#1} }}
\newcommand{\canononexG}[1]{\of{ \canononexfirst{#1}, \canononexsecondG{#1} }}
\newcommand{\canonone}[1]{\Mset{\canononex{#1} }_x}
\newcommand{\canononeG}[1]{\Mset{\canononexG{#1} }_{x\in V(G)} }

\title{On a Hierarchy of Spectral Invariants for Graphs\footnote{%
The results of this paper were presented in preliminary form in
the 41st International Symposium on Theoretical Aspects of Computer Science~\cite{stacs}}}
\author{V.~Arvind\thanks{The Institute of Mathematical Sciences (HBNI)
and Chennai Mathematical Institute, Chennai, India.}, 
Frank Fuhlbrück\thanks{Institut für Informatik, Humboldt-Universität zu Berlin, Germany.}, 
Johannes Köbler${}^\dagger$, Oleg Verbitsky${}^\dagger$\,\thanks{Supported by DFG grant KO 1053/8--2. 
On leave from the IAPMM, Lviv, Ukraine.}}

\date{}

\begin{document}

\maketitle

\begin{abstract}
  We consider a hierarchy of graph invariants that naturally extends
  the spectral invariants defined by Fürer (Lin.\ Alg.\ Appl.\ 2010) based on the angles formed
  by the set of standard basis vectors and their projections onto eigenspaces of the adjacency matrix.
  We provide a purely combinatorial characterization of this hierarchy
  in terms of the walk counts. This allows us to give a
  complete answer to Fürer's question about the strength
  of his invariants in distinguishing non-isomorphic graphs
  in comparison to the 2-dimensional Weisfeiler-Leman algorithm,
  extending the recent work of Rattan and Seppelt (SODA 2023).
  As another application of the characterization, we prove that
  almost all graphs are determined up to isomorphism in terms of the spectrum and
  the angles, which is of interest in view of the long-standing open problem
  whether almost all graphs are determined by their eigenvalues alone.
  Finally, we describe the exact relationship between the hierarchy and
  the Weisfeiler-Leman algorithms for small dimensions, as also some other
  important spectral characteristics of a graph such as the
  generalized and the main spectra. 
\end{abstract}

\section{Introduction}\label{s:intro}

The spectrum of a graph is a remarkable graph invariant that has found
numerous applications in computer science; e.g.,
\cite{HooryLW06,MoharP93}.  These applications are based on analyzing
relevant information contained in the eigenvalues of a given
graph. The maximum information possible is evidently obtained for
graphs that are determined by their spectra up to isomorphism. This
graph class, which is of direct relevance to the graph isomorphism
problem, is often abbreviated as \emph{DS} (for \emph{D}etermined by \emph{S}pectrum). Thus, a graph $G$ is DS if every
graph cospectral to $G$, i.e., with the same spectrum as $G$, is
actually isomorphic to $G$.  Though the problem of characterizing DS
graphs has been intensively studied since the beginning of spectral
graph theory (see \cite{CvetkovicRS10} and references therein), we are
still far from a satisfactory solution. In particular, a long-standing
open question \cite{HaemersS04,vanDamH03} is whether or not almost all
graphs are~DS. Here and in the rest of the paper, we say that
\emph{almost all} graphs have some property if a uniformly distributed
random $n$-vertex graph\footnote{In the Erd\H{o}s-R\'{e}nyi
  $\mathcal{G}(n,1/2)$ random graph model to be precise.} has this property with
probability approaching $1$ as $n$ goes to infinity.

Somewhat surprisingly at first sight, the area is connected to a purely combinatorial
approach to the graph isomorphism problem. In their seminal work, Weisfeiler and Leman \cite{WLe68}
proposed and studied a method for distinguishing a graph $G$ from another non-isomorphic graph
by computing a sequence of canonical partitions of $V(G)^2$ into color classes.
The final partition of $V(G)^2$ results in a \emph{coherent configuration},
a concept which is studied in algebraic combinatorics \cite{CP2019} and plays an important role
in isomorphism testing \cite{Babai16}.
The method of \cite{WLe68} is now called the 2-dimensional Weisfeiler-Leman algorithm (2-WL).
A similar approach based on partitioning $V(G)$ is known as \emph{color refinement}
and is often called 1-WL. Even this one-dimensional method is
quite powerful as it suffices for identification of almost all graphs \cite{BabaiES80}.
Here, a graph $G$ is said to be \emph{identifiable} by 1-WL if $G$ can be distinguished
from every non-isomorphic graph $H$ based on the output of 1-WL on $G$ and $H$.
The 2-WL algorithm is even more powerful, and construction of graphs not identifiable by 2-WL
is rather tricky. Particular examples of such graphs are based on rare combinatorial objects.
General constructions, like the far-reaching one in \cite{CaiFI92}, give rather
sporadic families of ``hard'' graphs. It turns out (see, e.g., \cite{BarghiP09,Fuerer10}) that if two graphs
are indistinguishable by 2-WL, then they are cospectral. As a consequence, graphs
not identifiable by~2-WL are examples of non-DS graphs.

Our overall goal in this paper is a systematic exploration of connections between
spectral and combinatorial approaches to finding expressive graph invariants.
A \emph{graph invariant} $\calI$ is a function of a graph such that
$\calI(G)=\calI(H)$ whenever $G$ and $H$ are isomorphic graphs.
An invariant $\calI$ is \emph{stronger} than invariant $\calI'$
if $\calI(G)$ determines $\calI'(G)$. That is, $\calI(G)=\calI(H)$ implies
$\calI'(G)=\calI'(H)$. Equivalently, we sometimes say that $\calI'$ is \emph{weaker}
than $\calI$ and write $\calI'\preceq\calI$. A stronger invariant can be more
effective in distinguishing non-isomorphic graphs: If $\calI'\preceq\calI$ and
$\calI'$ can distinguish non-isomorphic graphs $G$ and $H$, i.e.,
$\calI'(G)\ne\calI'(H)$, then these graphs are distinguishable by $\calI$ as well.

Let $\spec(G)$ denote the spectrum of a graph $G$, and $\wlk2(G)$ denote the
output of 2-WL on $G$ (a formal definition is given in Section~\ref{ss:wl}).
The discussion above shows that $\spec\preceq\wlk2$. This can be seen as evidence
of limitations of $\spec$, as well as evidence of the power of~$\wlk2$.

A reasonable question is what can be achieved if $\spec(G)$ is enhanced
by other spectral characteristics of the adjacency matrix of $G$.
One such line of research in spectral graph theory considers $\spec(G)$ augmented
with the multiset of all angles between the standard basis vectors and the eigenspaces
of $G$. The corresponding graph invariant, which along with \emph{E}igenvalues icludes
also \emph{A}ngles, is denoted by $\ea(G)$.
The parameters and properties of a graph $G$ which are
determined by $\ea(G)$
are called \emph{EA-reconstructible} and are thoroughly studied by Cvetkovi\'{c}
and co-authors; see \cite[Ch.~4]{CvetkovicRS97} and \cite[Ch.~3]{CvetkovicRS10}.

A further natural step is to take into consideration the multisets of angles between
the projections of the standard basis vectors onto eigenspaces.
Fürer \cite{Fuerer10} uses this additional data to define two new graph
invariants, namely, the \emph{weak} and \emph{strong spectral invariants}.
We denote these spectral invariants by $\weak$ and
$\strong$ respectively; formal definitions are in Section \ref{s:spec}.
Fürer shows that both $\weak$ and $\strong$ remain weaker than $\wlk2$. That is,
\begin{equation}
  \label{eq:fuerer}
\weak\preceq\strong\preceq\wlk2
\end{equation}
(note that $\spec\preceq\weak$ by definition).
An open problem posed in \cite{Fuerer10} is to determine  which of the relations in \refeq{fuerer}
are strict. Rattan and Seppelt, in their recent paper~\cite{RattanS23}, show that this small
hierarchy does not entirely collapse by separating $\weak$ and $\wlk2$.
Hence, at least one of the two relations in \refeq{fuerer}
is strict. Fürer \cite{Fuerer10} conjectures that the first relation
in \refeq{fuerer} is strict and does not exclude that the last two invariants in \refeq{fuerer}
are equivalent, and our aim is to give precise answers to these questions.

In~\cite{RattanS23} the invariants $\weak$ and $\wlk2$ are separated
by introducing a new natural graph invariant $\wlk{3/2}$, whose strength lies between
$\wlk1$ and $\wlk2$. The authors give an elegant algebraic characterization of $\wlk{3/2}$
using which they show that $\weak\preceq\wlk{3/2}$. The final step in their analysis
is an example of graph pair that separates $\wlk{3/2}$ and~$\wlk2$.

Our approach is different. First, we observe that the invariants $\weak$ and $\strong$
are part of a broader scheme, presented in Section \ref{s:frame}, that
leads to a potentially infinite hierarchy of graph invariants.
We define the corresponding spectral hierarchy containing $\weak$ and $\strong$ on its lower levels in Section \ref{s:spec}. Another level is taken by the aforementioned invariant $\ea$.
In Section \ref{s:char} we characterize this hierarchy in terms of walk counts.
A connection between spectral parameters and walk counts is actually well
known (see an overview in Subsection \ref{ss:background}). With a little extra effort
we are able to show that this connection is tight; see Theorem \ref{thm:char}.
This yields a purely combinatorial characterization of the invariants $\ea$, $\weak$, and $\strong$
(Corollary \ref{cor:char}), which also reveals some new relations.
For example, we notice that $\weak$ determines the \emph{generalized spectrum} of a graph (Theorem~\ref{thm:genspec-weak}).

As another application of our combinatorial characterization, we prove
that almost all graphs are determined up to isomorphism by $\weak$,
that is, by the eigenvalues and the angles formed by the standard
basis vectors and their projections onto eigenspaces (Corollary \ref{cor:aa}). We find this
interesting in the context of the open problem mentioned above:
whether or not almost all graphs are~DS.

We present the relations between the spectral and combinatorial invariants
under consideration in Section \ref{s:hierarchy}; see the diagram in Fig.~\ref{fig:diagram}.
In Section \ref{s:sep} we prove that this diagram is complete, that is,
it shows all existent relations, and all of these relations are strict
(perhaps up to higher levels of the hierarchy whose separation remains open).
In particular, both relations in \refeq{fuerer} are strict, which gives
a complete answer to Fürer's questions. Another noteworthy separation
is $\strong\npreceq\wlk{3/2}$ (Theorem \ref{thm:sep0}). Curiously,
the separating pair of graphs is the same that was used in \cite{RattanS23},
which yields more information now because we also use our characterization
in Theorem~\ref{thm:char}.

The more involved separations are shown in Theorems \ref{thm:sep}, \ref{thm:sep2}, and \ref{thm:lower-levels}.
The corresponding separating examples are not ad hoc. They are obtained by a quite general construction (Lemma \ref{lem:main}).
The construction is based on a considerable extension of the approach taken in \cite{ESA23} to separate various
concepts related to 1-WL and the walk matrix of a graph (an important notion discussed in Section \ref{ss:characterization}).
Implementation of the construction requires vertex-colored strongly regular graphs with
certain properties. The required colorings were found by a computer assisted search among
members of the family of strongly regular graphs on 25 vertices.

\section{Preliminaries}

\subsection{From isomorphism-invariant colorings to isomorphism invariants}\label{s:frame}

Let $\calC$ be a set of colors and $\chi\function{V(G)^2}{\calC}$ be a coloring
of vertex pairs in a graph $G$. It is natural to see $\chi(x,x)$ as the color
of a vertex $x$. We suppose that $\chi=\chi_G$ is defined for every graph $G$.
That is, speaking of a coloring $\chi$, we actually mean a map $G\mapsto\chi_G$.
Such a coloring $\chi$ is \emph{isomorphism invariant} if for every isomorphism
$f$ from a graph $G$ to a graph $H$ (the equality $G=H$ is not excluded)
we have $\chi_G(x,y)=\chi_H(f(x),f(y))$ for all vertices $x$ and $y$ in~$G$.

The simplest isomorphism-invariant colorings are the adjacency relation $A$
and the identity relation $I$. That is, $A(x,y)=1$ if $x$ and $y$ are adjacent
and $A(x,y)=0$ otherwise. For the identity relation, $I(x,y)=1$ if $x=y$
and $I(x,y)=0$ if $x\ne y$. Below, we will consider $A$ and $I$ also as
the adjacency and the identity matrices. Other examples of isomorphism-invariant colorings
are the distance $d(x,y)$ between two vertices $x$ and $y$ and the number
of all such triangles in the graph that contain the vertices $x$ and $y$.
One can consider also more complex definitions
like the triple $\chi(x,y)=(\deg x,d(x,y),\deg y)$, where $\deg x$ denotes the degree of~$x$.

Given an isomorphism-invariant coloring $\chi$, we can build on it to define various
graph invariants. The simplest such examples are
\begin{eqnarray*}
\calI_1(G)&=&\Mset{\chi(x,y)}_{x,y\in V(G)},\\
\calI_2(G)&=&\of{\Mset{\chi(x,x)}_{x\in V(G)},\Mset{\chi(x,y)}_{x,y\in V(G)}},\\
\calI_3(G)&=&\Mset{(\chi(x,x),\chi(x,y),\chi(y,y))}_{x,y\in V(G)},\\
\calI_4(G)&=&\canononeG{\chi},
\end{eqnarray*}
where $\Mset{\ldots}$ denotes a multiset.
Note that $\calI_1(G)\preceq\calI_2(G)\preceq\calI_3(G)\preceq\calI_4(G)$.

Given an isomorphism-invariant coloring $\chi$,
we can define a hierarchy of ever more complex graph invariants.
We first inductively define a sequence of colorings
$\chi_0,\chi_1,\chi_2,\ldots$ of single vertices by
\begin{equation}\label{eq:chi}
\chi_0(x)=\chi(x,x)\text{ and }\chi_{r+1}(x)=\of{\chi_r(x),\Mset{(\chi(x,y),\chi_r(y))}_{y\in V(G)}}.
\end{equation}
This definition is a natural extension of the well-known concept of \emph{color refinement (CR)}
to edge- and vertex-colored graphs. In broad outline, CR computes an isomorphism-invariant
color of each vertex in an input graph and recognizes two graphs as non-isomorphic
if one of the colors occurs in one graph more frequently than in the other.
In the case of an uncolored undirected graph $G$, CR starts with a uniform coloring
$\chi_0$ of $V(G)$, that is, $\chi_0(x)=\chi_0(x')$ for all $x,x'\in V(G)$.
In the $(r+1)$-th round, the preceding coloring $\chi_r$ is refined to a new coloring $\chi_{r+1}$.
For each vertex $x$, its new color $\chi_{r+1}(x)$ consists of $\chi_r(x)$ and
the multiset $\Mset{\chi_r(y)}_{y\in N(x)}$ of all colors occurring in the neighborhood $N(x)$ of $x$.
In other words, CR counts how frequently each $\chi_r$-color occurs among the vertices
adjacent to $x$ (or, equivalently, among the vertices non-adjacent to $x$).
In an edge- and vertex-colored graph $G$, each vertex pair $(x,y)$ is assigned a color,
which we denote by $\chi(x,y)$.
The edge colors must be taken into account while computing the refined color $\chi_{r+1}(x)$.
In the colored case, CR first splits all vertices $y$ into classes depending on $\chi(x,y)$
and then computes the frequencies of $\chi_r(y)$ within each class. This is exactly what \refeq{chi} does.

Note that
\begin{equation}
  \label{eq:one}
\chi_1(x)=\of{\chi(x,x),\Mset{(\chi(x,y),\chi(y,y))}_{y\in V(G)}}.
\end{equation}
In addition, we set
\begin{equation}
  \label{eq:half}
\chi_{1/2}(x)=\of{\chi(x,x),\Mset{\chi(x,y)}_{y\in V(G)}}
\end{equation}
Now, we define
\begin{equation}
  \label{eq:chir}
\chir r(G)=\Mset{\chi_r(x)}_{x\in V(G)}.
\end{equation}
For each $r$, the coloring $\chi_r$ is isomorphism invariant in the sense that
$\chi_r(f(x))=\chi_r(x)$ for any isomorphism of graphs $f$.
This readily implies that $\chir r$ is a graph invariant.
As easily seen, $\chir r\preceq\chir s$ if $r\le s$.

\subsection{First two dimensions of combinatorial refinement}\label{ss:wl}

We now give formal descriptions of the isomorphism tests 1-WL and 2-WL
 already introduced in Section \ref{s:intro}. Note that 1-WL is an alternative name
for CR. In what follows we will apply the 1-WL procedure to \emph{vertex-colored} graphs.
Given a vertex-colored graph $G$, we consider a coloring $\chi$ of $V(G)^2$
defined by $\chi(x,y)=A(x,y)$, i.e., according to the adjacency relation, for $x\ne y$
and by setting $\chi(x,x)$ to be the color of the vertex $x$. On an input $G$,
1-WL iteratively computes the vertex colorings $\chi_r$ according to \refeq{chi}.
After performing $n$ iterations, where $n$ is the number of vertices in $G$,
1-WL outputs $\wlk1(G)=\chir n(G)$ as defined by \refeq{chir}.
The coloring $\chi_n$ of $V(G)$ is sometimes referred to as \emph{stable}
because the color classes of $\chi_r$ stay the same for all $r\ge n$
(note that the color partition can stabilize even earlier).
A crucial observation is that if $\chir n(G)=\chir n(H)$ for another graph $H$, then
$\chir r(G)=\chir r(H)$ for all $r\ge n$.
Two graphs $G$ and $H$ are recognized as non-isomorphic if $\wlk1(G)\ne\wlk1(H)$.

2-WL can be similarly formulated, except that it computes colorings of
vertex pairs.  If an input graph $G$ is uncolored, then 2-WL begins
with an initial coloring $\chi_0$ of $V(G)^2$ defined by
$\chi_0(x,y)=A(x,y)$ if $x\ne y$ and by $\chi_0(x,x)=2$ for every
vertex $x$ of $G$.  If $G$ is a vertex-colored graph, then
$\chi_0(x,y)$ must include also the colors of $x$ and $y$.
Furthermore,
$$
\chi_{r+1}(x,y)=\of{
  \chi_{r}(x,y),\,\,
  \Mset{(\chi_{r}(x,z),\,\chi_{r}(z,y))}_{z\in V(G)}}.
$$
Thus, the new color of a pair $(x,y)$ can be seen as a kind of ``superposition'' of the old color pairs observable
along all extensions of $(x,y)$ to a triangle $xzy$. Finally, 2-WL outputs the multiset
$
\wlk2(G)=\Mset{\chi_{n^2}(x,y)}_{x,y\in V(G)}
$.

\section{A hierarchy of spectral invariants}\label{s:spec}

Speaking of an $n$-vertex graph $G$, we will assume that $V(G)=\{1,2,\ldots,n\}$.
Let
\begin{equation}
  \label{eq:mu}
\mu_1<\mu_2<\ldots<\mu_m
\end{equation}
be all pairwise distinct
eigenvalues of the adjacency matrix $A$ of $G$. Let $\spec(G)$ denote the spectrum of $G$,
i.e., the multiset of all eigenvalues where each $\mu_i$ occurs with its multiplicity.
As mentioned before, $\spec(G)$ is a well-studied graph invariant with
numerous applications in computer science.

Let $E_i$ be the eigenspace of $\mu_i$. Recall that $E_i$ consists of all
eigenvectors of $\mu_i$, i.e., $E_i=\Set{v\in\bR^n}{Av=\mu_iv}$.
Let $P_i$ be the matrix of the orthogonal projection of $\bR^n$ onto $E_i$.
Note that $P_i^2=P_i=P_i^\top$. For $1\le x,y\le n$, the matrix entry $P_i(x,y)$
can be seen a color of the vertex pair $(x,y)$. This coloring is
isomorphism invariant. Indeed, an isomorphism from $G$ to another $n$-vertex garph $H$
determines an isometry of $\bR^n$ which permutes the standard basis vectors and maps
the $\mu_i$-eigenspace of $G$ onto the $\mu_i$-eigenspace of~$H$.

Throughout the paper, we use the following notational convention
for compact representations of sequences.

\begin{notation}\label{not}
  For an indexed set $\Set{a_i}$ with index $i$ ranging through the interval of integers
  $s,s+1,\ldots,t-1,t$ we set $a_*=(a_s,\ldots,a_t)$.
\end{notation}

\noindent
In particular,
$$
P_*(x,y)=\of{P_1(x,y),\ldots,P_m(x,y)}.
$$
Since the order on the index set is determined by \refeq{mu},
$P_*$ is also an isomorphism-invariant coloring.
Following the general framework in Section \ref{s:frame}, the coloring $P_*$
determines the sequence of graph invariants $\ppr0,\ppr{1/2},\ppr1,\ppr2,\ldots$.
In particular, by \refeq{chi}--\refeq{chir} we have
\begin{eqnarray}
  \ppr0(G)&=&\Mset{\ P_*(x,x)\ }_{1\le x\le n},\label{eq:pp0}\\
  \ppr{1/2}(G)&=&\Mset{\  \of{\  P_*(x,x),\ \Mset{P_*(x,y)}_{1\le y\le n}\  }\  }_{1\le x\le n},\nonumber\\
  \ppr1(G)&=&\Mset{\  \of{\  P_*(x,x),\ \Mset{\of{P_*(x,y),P_*(y,y)}}_{1\le y\le n}\  }\  }_{1\le x\le n}.\nonumber
\end{eqnarray}

Fürer \cite{Fuerer10} introduces the \emph{weak}
and \emph{strong spectral invariants}. Using our notation, Fürer's spectral invariants (FSI)
can be defined as follows:
\begin{eqnarray}
  \weak(G)&=&\of{\spec(G),\ppr{1/2}(G)} \text{ and }\label{eq:F1}\\
  \strong(G)&=&\of{\spec(G),\ppr1(G)}.\label{eq:F2}
\end{eqnarray}

The entries of the projection matrices $P_i$ have a well-known geometric meaning \cite{CvetkovicRS10,CvetkovicRS93}.
For $1\le x\le n$, the standard basis vector $\mathrm{e}_x$ of $\bR^n$ has 1 in the position $x$
and 0 elsewhere. The \emph{angle} $\alpha_{i,x}$ of a graph $G$ is defined to be
the cosine of the angle between $\mathrm{e}_x$ and the eigenspace $E_i$, i.e.,
the angle between $\mathrm{e}_x$ and its projection $P_i\mathrm{e}_x$ onto~$E_i$.
We have the equality
\begin{equation}
  \label{eq:Pixx}
P_i(x,x)=\alpha_{i,x}^2.
\end{equation}
Indeed, let $\langle u,v\rangle$ denote the scalar product of two vectors $u,v\in\bR^n$. Then
$$
P_i(x,x)=\langle\mathrm{e}_x,P_i\mathrm{e}_x\rangle=\|\mathrm{e}_x\| \|P_i\mathrm{e}_x\| \alpha_{i,x}=\alpha_{i,x}^2.
$$
Furthermore, let $\alpha_{i,xy}$ be the cosine of the angle between the projections
$P_i\mathrm{e}_x$ and $P_i\mathrm{e}_y$ of the standard basis vector $\mathrm{e}_x$
and $\mathrm{e}_y$ onto $E_i$. If $\mathrm{e}_x$ or $\mathrm{e}_y$ is orthogonal to $E_i$,
i.e., $\alpha_{i,x}=0$ or $\alpha_{i,y}=0$, then the angle is undefined and we set $\alpha_{i,xy}=0$
in this case. In particular, $\alpha_{i,xx}=0$ if $\alpha_{i,x}=0$
while $\alpha_{i,xx}=1$ if $\alpha_{i,x}\ne0$.
Equality \refeq{Pixx} generalizes to
\begin{equation}
  \label{eq:Pixy}
P_i(x,y)=\alpha_{i,x}\alpha_{i,y}\alpha_{i,xy}.
\end{equation}
Indeed,
$$
P_i(x,y)=\langle\mathrm{e}_x,P_i\mathrm{e}_y\rangle=\langle\mathrm{e}_x,P_i^2\mathrm{e}_y\rangle
=\langle P_i\mathrm{e}_x,P_i\mathrm{e}_y\rangle=\|P_i\mathrm{e}_x\| \|P_i\mathrm{e}_y\| \alpha_{i,xy}=
\alpha_{i,x}\alpha_{i,y}\alpha_{i,xy}.
$$
Now, define a coloring $\alpha_i$ by $\alpha_i(x,x)=\alpha_{i,x}$ and $\alpha_i(x,y)=\alpha_{i,xy}$ for $x\ne y$.
This coloring is isomorphism invariant basically because an isomorphism is represented
by a permutation matrix, which is the transformation matrix of an isometry of~$\bR^n$.
Using Notation \ref{not},
$$
\alpha_*(x,y)=\of{ \alpha_1(x,y),\ldots,\alpha_m(x,y) },
$$
where $\alpha_*$ is an isomorphism-invariant coloring as well.
The corresponding graph invariants $\ui{\alpha_*}r$ are closely related to the invariants $\ppr r$.
More precisely, we say that two graph invariants $\calI$ and $\calI'$ are \emph{equivalent} and write $\calI\equiv\calI'$
if $\calI'\preceq\calI$ and $\calI\preceq\calI'$.

\begin{lemma}\label{lem:P-vs-a}
$\ppr r\equiv\ui{\alpha_*}r$ for every integer $r\ge0$.
\end{lemma}

\begin{proof}
  $r=0$. The equivalence $\Mset{P_*(x,x)}_x\equiv\Mset{\alpha_*(x,x)}_x$ readily follows from Equality~\refeq{Pixx}
  because the angles $\alpha_{i,x}$ are non-negative.

\bigskip

$r=1$. The equivalence reads
$$
\canonone{P_*}\equiv\canonone{\alpha_*}.
$$
For the ``$\preceq$'' part, it suffices to show for each $x$ that $\canononex{P_*}$ is determined by $\canononex{\alpha_*}$.
Indeed, as we already know, $P_*(x,x)$ is determined by $\alpha_*(x,x)$ due to~\refeq{Pixx}.
In order to show that $\canononexsecond{P_*}$ is determined by $\canononexsecond{\alpha_*}$,
we fix also $y$ and notice that $\canononexsubsecond{P_*}$ is determined by $\canononexsubsecond{\alpha_*}$
because for $1\leq i\leq m$
\begin{equation}
  \label{eq:Paaa}
P_i(x,y)=\alpha_i(x,x)\alpha_i(x,y)\alpha_i(y,y)
\end{equation}
as a consequence of Equality~\refeq{Pixy}.

The proof of the ``$\succeq$'' part is similar but requires a little bit more care.
Specifically, we have to argue that for fixed $x$ and $y$, $\canononexsubsecond{\alpha_*}$
is determined by $\canononexsubsecond{P_*}$. If $\alpha_i(x,x)=0$ or $\alpha_i(y,y)=0$,
then we cannot determine $\alpha_i(x,y)$ from \refeq{Paaa} by division but we know
that $\alpha_i(x,y)=0$ in this case just by definition.

In general, we proceed like above using induction with $r=0$ as the base case.
\end{proof}

Motivated by the equivalence $\ppr0\equiv\ui{\alpha_*}0$, we define the graph invariant $\ea$
similar to \refeq{F1}--\refeq{F2} as
\begin{equation}
  \label{eq:defEA}
\ea(G)=\of{ \spec(G), \ppr0(G) },
\end{equation}
where the abbreviation $\ea$ stands for \emph{E}igenvalues and \emph{A}ngles
and corresponds to the known concept \cite{CvetkovicRS97,CvetkovicRS10} mentioned in the introduction.

\section{Characterization of the spectral invariants by walk counts}\label{s:char}

A walk of length $k$ (or $k$-walk) from a vertex $x$ to a vertex $y$ is a sequence of vertices
$x=x_0,x_1,\ldots,x_k=y$ such that every two successive vertices $x_i,x_{i+1}$, for $0\le i<k$, are adjacent.
Let $w_k(x,y)$ denote the number of walks of length $k$ from $x$ to $y$ in a graph.
Obviously, $w_k$ is an isomorphism-invariant coloring in the sense of Section \ref{s:frame}.
In accordance with Notation \ref{not},
we also consider the isomorphism-invariant coloring $w_*$ defined by
$$
w_*(x,y)=\of{ w_0(x,y),w_1(x,y),\ldots,w_{n-1}(x,y) },
$$
where $n$, as usually, denotes the number of vertices in a graph.
Note that for each $y$, the matrix $\of{w_k(x,y)}_{1\le x\le n,\, 0\le k\le n-1}$
determines the value of $w_k(x,y)$ for every $x$ and for every arbitrarily large $k$.

It is well known that the walk counts are expressible in terms of spectral parameters
of a graph; see, e.g., \cite{CvetkovicRS10}. On the other hand, it is also well known
that the numbers of closed $k$-walks in a graph determine the graph spectrum.
We give a brief overview of these facts in Subsection \ref{ss:background}.
In Subsection \ref{ss:characterization} we make further use of this connection
between walks and spectra. We are able to characterize the spectral invariants
defined in Section \ref{s:spec} using solely the walk numbers, that is, in purely
combinatorial terms without involving any linear algebra.

\subsection{Linear-algebraic background and known relations}\label{ss:background}

Since the adjacency matrix $A$ of a graph $G$ is symmetric, the eigenspaces $E_i$
are pairwise orthogonal, and hence $P_iP_j=O$ for $i\ne j$, where $O$ denotes the zero matrix.
The spectral theorem for symmetric matrices says in essence that
$$
\bR^n=E_1\oplus\cdots\oplus E_m,
$$
that is, $\bR^n$ has an orthonormal basis consisting of eigenvectors of $A$.
This decomposition implies that
\begin{equation}
  \label{eq:I-decomp}
I=P_1+\cdots+P_m.
\end{equation}
From Equality \refeq{I-decomp} it is easy to derive the spectral decomposition
$$
A=\mu_1P_1+\cdots+\mu_mP_m.
$$
Raising both sides of this equality to the $k$-th power and taking into
account that $P_i^2=P_i$ and $P_iP_j=O$ for $i\ne j$, we conclude that
$$
A^k=\mu_1^kP_1+\cdots+\mu_m^kP_m.
$$
Since $w_k(x,y)=A^k(x,y)$, we get
\begin{equation}
  \label{eq:w-decomp}
w_k(x,y)=\mu_1^kP_1(x,y)+\cdots+\mu_m^kP_m(x,y).
\end{equation}

Let $c_k(G)=\sum_{x\in V(G)}w_k(x,x)$ denote the total number of closed $k$-walks in
a graph $G$. In particular, $c_0(G)=n$ and $c_1(G)=0$.

\begin{lemma}[folklore; see, e.g., \cite{GarijoGN11}]\label{lem:cosp}
 $\spec(G)=\spec(H)$ if and only if $c_k(G)=c_k(H)$ for $k=0,1,\ldots,n$.
\end{lemma}

\subsection{General characterization and its consequences}\label{ss:characterization}

\begin{theorem}\label{thm:char}
  $(\spec,\ppr r)\equiv\wwr r$ for every $r=0,1/2,1,2,\ldots$.
\end{theorem}

Before proving Theorem \ref{thm:char}, we describe some of its consequences.

\begin{corollary}\label{cor:char}\hfill

  \begin{enumerate}[\bf 1.]
  \item
    $\ea\equiv\wwr0$.
  \item
    $\weak\equiv\wwr{1/2}$.
  \item
    $\strong\equiv\wwr1$.
  \end{enumerate}
\end{corollary}

Parts 2 and 3 of Corollary \ref{cor:char}, which are particular cases of Theorem \ref{thm:char}
for $r=1/2$ and $r=1$ respectively, provide a characterization of both Fürer's spectral invariants.
We now comment on Part 1. By Definition \refeq{defEA}, this part is
the special case of Theorem \ref{thm:char} for $r=0$.
Note that $\wwr0(G)=\wwr0(H)$ if and only if the graphs $G$ and $H$ are \emph{closed-walk-equivalent}
in the sense that there is a bijection $f\function{V(G)}{V(H)}$
such that $w_k(x,x)=w_k(f(x),f(x))$ for all $x\in V(G)$ and all $k$.
On the other hand, let us say that  $G$ and $H$ are \emph{EA-equivalent}
if these graphs have the same eigenvalues and angles, i.e., $\ea(G)=\ea(H)$.
As seen from the summary in Subsection \ref{ss:background}, there
are well-known connections between the closed walk numbers
and the eigenvalues and angles. Part 1 of Corollary \ref{cor:char}
pinpoints the fact that the two equivalence concepts actually coincide.

Corollary \ref{cor:char} reveals connections of spectral invariants to other graph
invariants studied in the literature, which we introduce now.

The \emph{walk matrix} $W$ of a graph $G$ is indexed by vertices $1\le x\le n$
and the length parameter $0\le k\le n-1$ and defined by
\begin{equation}
  \label{eq:Wxk}
W(x,k)=\sum_{y=1}^nw_k(x,y).
\end{equation}
That is, $W(x,k)$ is the total number of $k$-walks starting from the vertex $x$.
Let $\kwlout1^k(G,x)$ denote the color assigned by $1$-WL to a vertex $x$
of $G$ after the $k$-th refinement round.
For a vertex $x$ of $G$, let $G_x$ denote the version of $G$ with $x$ \emph{individualized}.
This means that $G_x$ is a vertex-colored graph where $x$ has a special unique color
while the other vertices are colored uniformly.
We now state two well-known facts.

\begin{lemma}\label{lem:w-vs-cr}\hfill

\begin{enumerate}[\bf 1.]
\item
  $W(x,k)$ is determined by $\kwlout1^k(G,x)$;
\item
  $w_k(x,y)$ is determined by $\kwlout1^k(G_x,y)$;
\end{enumerate}
\end{lemma}

Part 1 of Lemma \ref{lem:w-vs-cr} is proved in algebraic terms in \cite[Theorem 2]{PowersS82}.
Another proof, involving logical concepts, is provided in \cite[Lemma~4]{Dvorak10}
and a direct combinatorial proof can be found in \cite[Lemma~8]{ESA23}.
Part 2 is a straightforward extension of Part~1.

In addition to the graph invariants $\wlk1$ and $\wlk2$ introduced in Section \ref{ss:wl}, we define
$$
\wlk{3/2}(G)=\Mset{\wlk1(G_x)}_{x\in V(G)}.
$$
Note that different $x$ and $x'$ are individualized in $G_x$ and $G_{x'}$ by the same color.
This yields a chain of graph invariants $\wlk d$ for $d\in\{1,3/2,2\}$, where $\wlk c\preceq\wlk d$ if $c\le d$.
The walk matrix naturally gives us a graph invariant, which we denote by $\wm$
and define as $\wm(G)=\Mset{W(x,*)}_{x\in V(G)}$ where $W(x,*)=\big(W(x,0),W(x,1),\ldots,\allowbreak W(x,n-1)\big)$.
In other words, $\wm(G)$ is the multiset of the rows of the walk matrix of $G$.
Part 1 of Lemma \ref{lem:w-vs-cr} readily implies that $\wm\preceq\wlk1$. Thus,
$$
\wm\preceq\wlk1\preceq\wlk{3/2}\preceq\wlk2.
$$

The second relation in the following corollary is the recent result in \cite{RattanS23}
already described in Section~\ref{s:intro}.

\begin{corollary}\label{cor:rel}
  $\wm\preceq\weak\preceq\wlk{3/2}$.
\end{corollary}

\begin{proof}
  By Part 2 of Corollary \ref{cor:char}, it is enough to show that
  $$
\wm\preceq\wwr{1/2}\preceq\wlk{3/2}.
$$
The former relation follows directly from the definitions of $\wm$ and $\wwr{1/2}$.
Indeed, $\wwr{1/2}(G)$ comprises the multiset $\Mset{w_*(x,y)}_{y\in V(G)}$ for each
vertex $x$ of $G$; see \refeq{half}. This multiset allows us to calculate the sum~\refeq{Wxk} for each $k$.
The latter relation follows from the definitions of $\wwr{1/2}$ and $\wlk{3/2}$
by Part 2 of Lemma~\ref{lem:w-vs-cr}.
\end{proof}

The relationship between the graph invariants discussed above is summarized
in Figure \ref{fig:diagram} below, which also puts these invariants in a somewhat
broader context.

The next consequence of Theorem \ref{thm:char} is interesting in view of the long-standing open question
whether almost all graphs are determined up to isomorphism by their spectrum~\cite{HaemersS04,vanDamH03}.

\begin{corollary}\label{cor:aa}
 Almost all graphs are determined up to isomorphism by $\weak$.
\end{corollary}

\begin{proof}
  It is known that if the walk matrix is non-singular, then it determines the adjacency
  matrix \cite{LiuS22}. Moreover, the walk matrix of a random graph is non-singular
  with high probability \cite{ORourkeT16}. As a consequence, almost all graphs are determined
  by the graph invariant $\wm$; see \cite[Theorem~7.2]{LiuS22}. The same is true also for $\weak$
  because $\weak$ is stronger than $\wm$ by Corollary~\ref{cor:rel}.
\end{proof}

Since $\spec\preceq\ea\preceq\weak$, a natural further question is
whether Corollary \ref{cor:aa} can be improved to show that almost all
graphs can be identified by the graph invariant $\ea$, that is, by
using the eigenvalues and the angles between each standard basis
vector and its projections onto eigenspaces but not between the
projections themselves.  If true, this would be yet closer to the
aforementioned open problem.

\subsection{Proof of Theorem \ref{thm:char}}

\textbf{The case of $r=0$.} We have to prove that
\begin{equation}
  \label{eq:r0}
\of{ \spec, \Mset{P_*(x,x)}_x } \equiv \Mset{w_*(x,x)}_x.
\end{equation}
The ``$\succeq$'' part immediately follows from Equality \refeq{w-decomp}.
In the other direction, the relation $\spec \preceq \Mset{w_*(x,x)}_x$ is a direct
consequence of Lemma \ref{lem:cosp}. To complete the proof of \refeq{r0},
we show that for each vertex $x$, the sequence $P_*(x,x)$ can be obtained
from the sequences $w_*(x,x)$ and $\mu_*$. To this end, put $y=x$ in Equality \refeq{w-decomp},
obtaining
\begin{equation}
  \label{eq:xx}
\mu_1^kP_1(x,x)+\cdots+\mu_m^kP_m(x,x)=w_k(x,x).
\end{equation}
This equality makes sense also for $k=0$. In this case, it reads
\begin{equation}
  \label{eq:xx0}
P_1(x,x)+\cdots+P_m(x,x)=1,
\end{equation}
which is true by Equality \refeq{I-decomp}
(if $\mu_i=0$, we need to use the convention $0^0=1$).
Consider Equalities \refeq{xx} for $k=0,1,\ldots,m-1$ as a system
of $m$ linear equations for $m$ unknowns $P_1(x,x),\ldots,P_m(x,x)$.
The coefficients of this system are powers of the $m$ pairwise distinct eigenvalues.
They form a Vandermonde matrix.
Therefore, the system is uniquely solvable, and the sequence $P_*(x,x)$ is determined.

\bigskip

\noindent
\textbf{The case of $r=1/2$.} Now we have to prove that
\begin{equation}
  \label{eq:r12}
  \of{ \spec, \Mset{ \of{ P_*(x,x), \Mset{P_*(x,y)}_y } }_x }
  \equiv
\Mset{ \of{ w_*(x,x), \Mset{w_*(x,y)}_y } }_x.
\end{equation}
The ``$\succeq$'' part is again an immediate consequence of Equality \refeq{w-decomp}.

Let us prove the part ``$\preceq$''. That is, we have to show that for a given graph,
the left hand side of \refeq{r12} can be obtained from the right hand side.
The spectrum is, as already observed in the case of $r=0$, determined by
the multiset $\Mset{w_*(x,x)}_x$, which is easy to obtain from the right hand side of \refeq{r12}.
Thus, in what follows we can assume that the sequence $\mu_1,\ldots,\mu_m$ of distinct eigenvalues is known.
For each vertex $x$, we have to compute the sequence $P_*(x,x)$ and the multiset $\Mset{P_*(x,y)}_y$.
The former task is solvable exactly as in the case of $r=0$, and we focus on the latter task.
In addition to $x$, we now fix also $y$ and consider Equalities \refeq{w-decomp}
for $k=0,1,\ldots,m-1$. Here, the equality for $k=0$ is actually a particular instance
of Equality \refeq{I-decomp}, that is, this is \refeq{xx0} if $y=x$ and
$$
P_1(x,y)+\cdots+P_m(x,y)=0
$$
if $y\ne x$. As in the case of $r=0$, for $m$ unknowns $P_1(x,y),\ldots,P_m(x,y)$
we obtain a system of $m$ linear equation whose coefficients form a non-singular Vandermonde matrix.
Hence, we can determine the sequence $P_*(x,y)$, completing the proof of~\refeq{r12}.

\bigskip

\noindent
\textbf{The case of $r\ge1$.} We proceed as above using induction.
To facilitate the notation, we set $\pi(x,y)=P_*(x,y)$ and $\omega(x,y)=w_*(x,y)$.
We write $a\mapsfrom b$ to say that $a$ is obtainable from $b$. As we have already seen,
\begin{equation}
  \label{eq:wfromsp}
  \omega(x,y)\mapsfrom\spec,\,\pi(x,y)
\end{equation}
by Equality \refeq{w-decomp} and
\begin{equation}
  \label{eq:pfromsw}
  \pi(x,y)\mapsfrom\spec,\,\omega(x,y)
\end{equation}
by the Vandermonde matrix argument.
The vertex colorings $\pi_r$ and $\omega_r$ are defined as in \refeq{chi}.
For every $r$,
\begin{equation}
  \label{eq:sfromw}
  \spec\mapsfrom\Mset{\omega_0(x)}_x\mapsfrom\Mset{\omega_r(x)}_x.
\end{equation}
The former relation is, as already observed above, a consequence of Lemma \ref{lem:cosp},
while the latter relation follows directly from the definition of $\omega_r$.
Therefore, in order to prove that
$$
\of{\spec,\Mset{\pi_r(x)}_x}\preceq\Mset{\omega_r(x)}_x,
$$
it suffices to prove for each $x$ that
\begin{equation}
  \label{eq:rpfromsw}
  \pi_r(x)\mapsfrom\spec,\,\omega_r(x).
\end{equation}
In order to prove that
$$
\Mset{\omega_r(x)}_x\preceq\of{\spec,\Mset{\pi_r(x)}_x},
$$
it suffices to prove for each $x$ that
\begin{equation}
  \label{eq:rwfromsp}
  \omega_r(x)\mapsfrom\spec,\,\pi_r(x).
\end{equation}
We prove \refeq{rpfromsw} and \refeq{rwfromsp} by induction on~$r$.

In the base case we have $r=0$.
The relation \refeq{rwfromsp} for $r=0$ follows from the relation \refeq{wfromsp} for $y=x$.
Similarly, the relation \refeq{rpfromsw} for $r=0$ follows from the relation \refeq{pfromsw} for $y=x$.

For the induction step, suppose that $r\ge1$. Consider first \refeq{rpfromsw}. Recall that
$$
\pi_{r}(x)=\of{\pi_{r-1}(x),\Mset{(\pi(x,y),\pi_{r-1}(y))}_{y}}.
$$
By the induction hypothesis,
$$
\pi_{r-1}(x)\mapsfrom\spec,\,\omega_{r-1}(x)\mapsfrom\spec,\,\omega_{r}(x).
$$
The latter relation follows from the fact that $\omega_{r-1}(x)$ is a part of $\omega_{r}(x)$.
Another part of $\omega_{r}(x)$ gives us the multiset $\Mset{(\omega(x,y),\omega_{r-1}(y))}_{y}$.
Therefore, it suffices to argue that, for each $y$,
$$
(\pi(x,y),\pi_{r-1}(y))\mapsfrom\spec,\,(\omega(x,y),\omega_{r-1}(y)).
$$
Indeed, $\pi(x,y)$ is determined by \refeq{pfromsw}, and
$\pi_{r-1}(y)$ is determined by the induction hypothesis.

The argument for \refeq{rwfromsp} is virtually the same, with the roles of $\pi$ and $\omega$
interchanged. In place of \refeq{pfromsw}, we have to refer to~\refeq{wfromsp}.
The proof is complete.

\section{The Hierarchy}\label{s:hierarchy}

Figure~\ref{fig:diagram} shows the invariants from Section \ref{ss:characterization} and
the relations between them as a part of a more general picture involving
also some other spectral invariants studied in the literature, which we
introduce in the next two subsections.

Moreover, we define a new invariant $\wwr{\bullet}$ which is, in a sense,
the limit of the sequence of invariants $\wwr r$ for $r=1,2,\ldots$.
The definition is rather general.
As we already mentioned, the formal framework of Section \ref{s:frame}
is analogous to the concept of color refinement. Given an isomorphism-invariant
coloring $\chi$, we therefore can along with graph invariants $\chir r$
define the stable version $\chir{\bullet}$. One possibility to do this
is to set $\chir{\bullet}(G)=\chir n(G)$, where $n$ is the number of vertices in $G$.
Note that $\chir r\preceq\chir{\bullet}$ for every~$r$.
For $\chi=\omega_*$, we obtain the potentially infinite chain
$$
\wwr0\preceq\wwr{1/2}\preceq\wwr1\preceq\wwr2\preceq\cdots\preceq\wwr{\bullet}.
$$
Note also that
$$
\wlk1\preceq\wwr{\bullet}\preceq\wlk2.
$$
The former relation is true because the inequality $\wlk1(G)\ne\wlk1(H)$ implies
the inequality $\wwr{\bullet}(G)\ne\wwr{\bullet}(H)$ (by taking into account that
the color $\omega_*(x,y)$ determines whether $x$ and $y$ are equal and whether they are adjacent).
To see the latter relation, we first recall the well-known fact
that $w_*(x,y)$ is determined by the color assigned to the vertex pair
$(x,y)$ by 2-WL. Furthermore, since $\wwr{\bullet}(G)$ is obtained
from $G$ endowed with the coloring $w_*$ by running the version of
1-WL for edge-colored graphs, the outcome can be simulated by~2-WL.

\begin{figure}
  \centering
  \begin{tikzcd}
\wlk2\rar\ar[d]&\wwr{\bullet}\rar\ar[ldd]&\cdots\rar&\wwr r\dar&(r\ge1)\hspace{6mm}\mbox{}\\
\wlk{3/2}\ar[d]\ar[rrrdd]&&&\vdots\dar&\\
\wlk1\ar[dd]&&&\wwr1\dar\rar&\lar\strong\\
&&&\wwr{1/2}\ar[llld]\ar[ld]\dar\rar&\lar\weak\hspace{1mm}\mbox{}\\
\wm\dar&&\genspec\ar[lld]\ar[rd]&\wwr0\dar\rar&\lar\ea\hspace{10mm}\mbox{}\\
\mea&&&\spec&
\end{tikzcd}
\caption{Relations between graph invariants. An arrow $\calI\to\calI'$ means $\calI'\preceq\calI$.
}
\label{fig:diagram}
\end{figure}

\subsection{Main eigenvalues and angles}

Let $\jj$ denote the all-ones vector (the dimension should be clear from the context).
Using our usual notation, suppose that $G$ has $m$ distinct eigenvalues $\mu_1,\ldots,\mu_m$,
and let $E_1,\ldots,E_m$ be
the corresponding eigenspaces of $G$. Consider the angle between $E_i$ and $\jj$ and denote
its cosine by $\beta_i$. If $\beta_i\ne0$, then the corresponding eigenvalue $\mu_i$ is called
a \emph{main eigenvalue}, and then the positive number $\beta_i$ is called a \emph{main angle}.
Let $\nu_1,\ldots,\nu_{m'}$ be the sequence of all main eigenvalues in the ascending order
and $\theta_1,\ldots,\theta_{m'}$ be the sequence of the main angles in the corresponding order.
We define a graph invariant $\mea$ by
$$
\mea(G)=\of{\nu_*,\theta_*}.
$$
A characterization of $\mea$ in terms of walk numbers is known.
We state it in the proposition below and, for the sake of completeness, provide
some proof details apparently missing in the available literature (the backward direction).
Let
$$
w_k(G)=\sum_{x\in V(G)}w_k(x)
$$
be the total number of $k$-walks in~$G$.
By $W_G$ we will denote the corresponding generating function, that is,
the formal series
$$
W_G(z)=\sum_{k=0}^\infty w_k(G)z^k.
$$

\begin{proposition}[folklore, e.g.~\cite{StackE}]\label{prop:stack}
  Let $G$ and $H$ be graphs with $n$ vertices. Then
  $\mea(G)=\mea(H)$ if and only if $w_k(G)=w_k(H)$ for $k=1,\ldots,n-1$.
\end{proposition}

\begin{proof}
  In the forward direction the proposition follows from the equality
  \begin{equation}
    \label{eq:mu-beta}
  w_k(G)=n\sum_{i=1}^{m'} \nu_i^k\theta_i^2
  \end{equation}
 for all $k\ge0$; see \cite[Theorem~1.3.5]{CvetkovicRS10} for details.  In particular, this
 equality holds true for $k=0$ because $w_0(G)=n$ and because
 $\sum_{i=1}^{m'}\theta_i^2=1$. If $k=0$ and $\nu_i=0$ for
 some $i$, then the sum in \refeq{mu-beta} involves the coefficient
 $0^0=1$. Thus, the equality $\mea(G)=\mea(H)$ implies that $w_k(G)=w_k(H)$ for all $k\ge0$.

 For the other direction,
 consider the generating function $W_G(z)$.
 This is a formal power series, i.e., an element of the ring $\bR[[z]]$.
 In this ring we have
   \begin{equation}
    \label{eq:WGz}
    W_G(z)=\sum_{k=0}^\infty \of{n\sum_{i=1}^{m'} \nu_i^k\theta_i^2}z^k =
    \sum_{i=1}^{m'} n\theta_i^2\sum_{k=0}^\infty(\nu_iz)^k=\sum_i\frac{n\theta_i^2}{1-\nu_iz};
  \end{equation}
  cf.~\cite[Eq.~(9)]{Cvetkovic78}.
  Note that all divisions are well defined because $1-\nu_iz$ are
  invertible elements of $\bR[[z]]$.
  If 0 is a main eigenvalue, then we suppose that $\nu_0=0$, where 
  the main eigenvalues in the rightmost expression of \refeq{WGz} are now indexed starting with $0$.
  Otherwise it is still convenient to set
  $\nu_0=0$ and $\theta_0=0$.  In either case, Equality \refeq{WGz} can
  be rewritten as
  $$
    W_G(z)=n\theta_0^2+\sum_{i\ne0}\frac{n\theta_i^2/\nu_i}{1/\nu_i-z}.
  $$
  From here we see that the numbers $\nu_i$ and $\theta_i$
  are uniquely determined by $W_G(z)$. This follows from the simple observation
  that if $c_0+\sum_{i\ne0} c_i/(a_i-z)=0$ in $\bR[[z]]$ for finitely many pairwise distinct reals $a_i$,
  then $c_i=0$ for all~$i$.
\end{proof}

As a direct consequence of Proposition \ref{prop:stack}, we get the relation $\mea\preceq\wm$.

\subsection{The generalized spectrum}

Another important spectral invariant of $G$ is the spectrum
of the complement graph $\compl G$. The equalities $\spec(G)=\spec(H)$
and, simultaneously, $\spec(\compl G)=\spec(\compl H)$ are equivalent to the condition that
the graphs $G$ and $H$ have the same \emph{generalized spectrum}.
For the definition of this concept and its various characterizations
we refer the reader to \cite{JohnsonN80} and \cite[Theorem~3]{vanDamHK07}.
We define the graph invariant $\genspec$ by
$$
\genspec(G)=\of{\spec(G),\spec(\compl G)}.
$$
We note that
\begin{equation}
  \label{eq:mea-genspec}
\mea\preceq\genspec.  
\end{equation}
This follows from Proposition \ref{prop:stack} and the following result in \cite{Cvetkovic78}.
Let $P_G$ denote the characteristic polynomial of a graph~$G$.

\begin{proposition}[Cvetkovi{\'c} \cite{Cvetkovic78}]\label{prop:cvetk}
$W_G(z)=\frac1z\of{(-1)^n\frac{P_{\compl G}(-1/z-1)}{P_G(1/z)}-1}$.
\end{proposition}

\noindent
Alternatively, \refeq{mea-genspec} can be obtained by using \cite[Theorem~3]{vanDamHK07}.

To complete the diagram in Figure \ref{fig:diagram}, it remains to prove the following relation.

\begin{theorem}\label{thm:genspec-weak}
  $\genspec\preceq\weak$.
\end{theorem}

To give the reader a flavor of how our characterization of spectral invariants
in terms of walk counts (Theorem \ref{thm:char}) works, we provide two proofs
of Theorem \ref{thm:genspec-weak}, one combinatorial that uses the characterization
of $\weak$ in Corollary \ref{cor:char} and another one, more algebraic, that uses the
original definition of~$\weak$.

\begin{proof}[Combinatorial proof of Theorem \ref{thm:genspec-weak}]
  The spectrum of $G$, hence also the characteristic polynomial $P_G$, is
  determined by $\weak(G)$ just by definition. We have to show that $\spec(\compl G)$
  or, equivalently, $P_{\compl G}$ is also determined. By Proposition \ref{prop:cvetk},
  $P_{\compl G}$ is obtainable from $P_G$ and $W_G$. Using Part 2 of Corollary \ref{cor:char},
  it remains to notice that $W_G$ is determined by $\wwr{1/2}(G)$. Indeed,
  $$
w_k(G)=\sum_x w_k(x)=\sum_x\sum_y w_k(x,y),
$$
where the right hand side is obtainable from the multiset $\Mset{\Mset{w_k(x,y)}_y}_x$,
which is a part of~$\wwr{1/2}(G)$.
\end{proof}

The algebraic proof is based on the fact that $\weak$ determines not only the main angles
but also the entire sequence of angles $\beta_i$ between the all-ones vector $\jj$ and the eigenspaces
ordered by the ascending order of the corresponding eigenvalues.

\begin{lemma}\label{lem:beta-weak}
  $\weak$ determines $\beta_*$.
\end{lemma}

\begin{proof}
Note that
  $$
\langle\jj,P_i\,\jj\rangle=\|\jj\| \|P_i\,\jj\| \beta_i.
$$
On the other hand,
$$
\langle\jj,P_i\,\jj\rangle=\langle\jj,P_i^2\,\jj\rangle=\langle P_i\,\jj,P_i\,\jj\rangle=\|P_i\,\jj\|^2,
$$
from which we conclude that
$\beta_i=\|P_i\,\jj\|/{\sqrt n}$.
Furthermore,
$$
\|P_i\jj\|^2=\langle\jj,P_i\jj\rangle=\Big\langle\sum_x\mathrm{e}_x,P_i\sum_y\mathrm{e}_y\Big\rangle
=\sum_{x,y}\langle\mathrm{e}_x,P_i\mathrm{e}_y\rangle=\sum_{x,y}P_i(x,y).
$$
It follows that $\|P_i\jj\|$ and, hence, $\beta_i$ can be determined from $\weak$
because $\weak$ gives us the multiset $\Mset{\Mset{P_i(x,y)}_y}_x$ for each $i=1,\ldots,n$.
\end{proof}

\begin{proof}[Another proof of Theorem \ref{thm:genspec-weak}]
  As it readily follows from Lemma \ref{lem:beta-weak}, $\mea(G)$ is determined by $\weak(G)$.
  By Proposition \ref{prop:stack}, this implies that $\weak(G)$ determines $W_G$.
  Therefore, $\spec(\compl G)$ is obtainable from $\weak(G)$ by Proposition~\ref{prop:cvetk}.
\end{proof}

\section{Separations}\label{s:sep}

Fürer \cite{Fuerer10} poses the open problem
of determining which of the relations in the chain
\begin{equation}
  \label{eq:fuerer2}
\weak\preceq\strong\preceq\wlk2
\end{equation}
are strict. As mentioned before, Rattan and Seppelt \cite{RattanS23} show that this chain 
does not entirely collapse. They separate $\weak$ and $\wlk2$ by proving that $\weak\preceq\wlk{3/2}$
and separating $\wlk{3/2}$ and $\wlk2$. Hence, at least one of the two relations in \refeq{fuerer2}
is strict. Fürer \cite{Fuerer10} conjectures that the first relation
in \refeq{fuerer2} is strict and does not exclude the possibility that the last two invariants in \refeq{fuerer2}
are equivalent. We settle this by showing that, in fact, both relations in \refeq{fuerer2}
are strict. We actually prove much more: up to one question remaining open,
the diagram shown in Figure \ref{fig:diagram} is exact in the sense that all present arrows
are non-reversible and that any two invariants not connected by arrows are provably incomparable.
The only remaining question concerns the chain $\wwr{\bullet}\to\cdots\to\wwr r\to\cdots\to\wwr1$;
see Problem \ref{probl:main} stated below and an approach to its solution in Theorem~\ref{thm:lower-levels}.

We now present a minimal set of separations from which 
all other separations follow. For compatibility with Figure \ref{fig:diagram},
we write $\calI\nrightarrow\calI'$ to negate $\calI'\preceq\calI$.

\begin{description}
\item[$\wlk1\nrightarrow\spec$.]
  To show this, we have to present two $\wlk1$-equivalent but not cospectral graphs.
  The simplest pair of $\wlk1$-equivalent graphs,
  $C_6=\raisebox{-1.6mm}{\begin{tikzpicture}[every node/.style=vertex,rotate=30,scale=.25]
\path (0:1cm) node (a) {}
      (60:1cm) node (b)  {} edge (a) 
      (120:1cm) node (c)  {} edge (b)
      (180:1cm) node (d)  {} edge (c)
      (240:1cm) node (e)  {} edge (d)
      (300:1cm) node (f)  {} edge (e) edge (a);
    \end{tikzpicture}}$
  and
  $2C_3=\raisebox{-1.6mm}{\begin{tikzpicture}[every node/.style=vertex,rotate=30,scale=.25]
\path (0:1cm) node (a) {}
      (60:1cm) node (b)  {} edge (a) 
      (120:1cm) node (c)  {} edge (b) edge (a)
      (180:1cm) node (d)  {} 
      (240:1cm) node (e)  {} edge (d)
      (300:1cm) node (f)  {} edge (e) edge (d);
    \end{tikzpicture}}$,
  works.
  Indeed, the eigenvalues of $C_n$ are $2\cos\frac{2\pi k}n$ for $k=0,1,\ldots,n-1$,
  and the spectrum of the disjoint union of graphs is the union of their spectra;
  see, e.g., \cite[Example 1.1.4 and Theorem 2.1.1]{CvetkovicRS10}.
\item[$\ea\nrightarrow\mea$.]
  It is known \cite{CvetkovicL97} that among trees with up to 20 vertices there is a single pair of non-isomorphic trees
  $T$ and $T'$,
  with 19 vertices, having the same eigenvalues and angles.
A direct computation shows that the total number of walks of length 9 is 35506 for $T$ and 35498 for $T'$.
By Proposition \ref{prop:stack}, the trees $T$ and $T'$ are not $\mea$-equivalent.
\item[$\genspec\nrightarrow\wm$.]
  The smallest, with respect to the
  number of vertices and the number of edges, pair of generalized
  cospectral graphs consists of 7-vertex graphs
  $G=\raisebox{-1.35mm}{\begin{tikzpicture}[every node/.style=vertex,scale=.25]
\path (0:1cm) node (a) {}
      (60:1cm) node (b)  {} edge (a) 
      (120:1cm) node (c)  {} edge (b)
      (180:1cm) node (d)  {} edge (c)
      (240:1cm) node (e)  {} edge (d)
      (300:1cm) node (f)  {} edge (e) edge (a)
      (0,0) node (central) {};
    \end{tikzpicture}}$
  and
  $H=\raisebox{-1.2mm}{\begin{tikzpicture}[every node/.style=vertex,rotate=30,scale=.25]
      \path (0,0) node (central) {}
      (0:.5cm) node (a) {} edge (central) 
      (0:1cm) node (aa) {} edge (a) 
      (120:.5cm) node (b) {} edge (central) 
      (120:1cm) node (bb) {} edge (b) 
      (240:.5cm) node (c) {} edge (central) 
      (240:1cm) node (cc) {} edge (c);
    \end{tikzpicture}}$; see \cite[Fig.~4]{HaemersS04}.
  Since $G$ has
  an isolated vertex and $H$ does not, these graphs are not
  $\wm$-equivalent.
\item[$\genspec\nrightarrow\ea$.]
  The same pair of graphs $G$ and $H$
  works. They are not $\ea$-equivalent because connectedness of a
  graph is determined by its spectrum and angles
  \cite[Theorem~3.3.3]{CvetkovicRS10}.  Another reason for this is that
  every graph with at most 9 vertices is determined by its spectrum
  and angles up to isomorphism~\cite{CvetkovicL97}.
\item[$\strong\nrightarrow\wlk1$.] An even stronger fact, namely
  $\wwr2\nrightarrow\wlk1$ is proved as Theorem \ref{thm:sep}
  below. The separation implies that $\strong\nrightarrow\wlk2$,
  answering one part of Fürer's question. Another consequence is also
  the separation $\wm\nrightarrow\wlk1$, which follows as well
  from~\cite[Theorem 3]{ESA23}.
\item[$\wlk1\nrightarrow\wlk{3/2}$.]
Consider $C_6$ and $2C_3$.
\item[$\wlk{3/2}\nrightarrow\strong$.]
  This is Theorem \ref{thm:sep0} below.
  As a consequence, $\weak\nrightarrow\strong$, answering the other part of Fürer's question.
  Another consequence is the separation $\wlk{3/2}\nrightarrow\wlk2$ shown in~\cite{RattanS23}.
\item[$\wwr{\bullet}\nrightarrow\wlk{3/2}$.]
   This is Theorem \ref{thm:sep2} below.
   It considerably strengthens the negative answer to Fürer's question by showing that
   the whole hierarchy of the invariants $\wwr r$ for all $r$ is strictly weaker than~$\wlk2$.
\end{description}

\begin{remark}
  Our results about the hierarchy in Figure \ref{fig:diagram} can also be useful for classification
  of other graph invariants definable through spectral concepts and combinatorial refinement.
  As an example, consider the graph invariant formed by the pair $(\wlk1,\spec)$.
  We thank the anonymous referee who drew our attention to a natural characterization
  of this invariant in terms of linear-algebraic expressions involving the adjacency matrix,
  the all-ones vector, matrix multiplication, trace, conjugation, and pointwise vector multiplication
  \cite[Prop.~8.3]{Geerts21}.
  Note that $\genspec\preceq(\wlk1,\spec)$ as a consequence of Proposition \ref{prop:cvetk}
  and the relation $\wm\preceq\wlk1$. The separation $(\wlk1,\spec)\nrightarrow\wlk{3/2}$
  immediately follows from the aforementioned separation $\wwr{\bullet}\nrightarrow\wlk{3/2}$.
  Using the example of $(\wlk1,\spec)$-equivalent graphs
  $G=\raisebox{-1.6mm}{\begin{tikzpicture}[every node/.style=vertex,scale=.25]
  \path (0,0) node (c) {}
      (0,1) node (t) {} edge (c)
      (0,-1) node (b) {} edge (c)    
      (-1,1) node (lt) {} edge (t)
      (1,1) node (rt) {} edge (t)
       (-1,-1) node (lb) {} edge (b) edge (lt)
       (1,-1) node (rb) {} edge (b)  edge (rt)
       (-1.5,0) node (l) {} edge (lt)  edge (lb)
       (1.5,0) node (r) {} edge (rt)  edge (rb);
     \end{tikzpicture}}$
   and
  $H=\raisebox{-1.6mm}{\begin{tikzpicture}[every node/.style=vertex,scale=.25]
  \path (0,0) node (c) {}
      (0,1) node (t) {}
      (0,-1) node (b) {}    
      (-1.5,1) node (lt) {} edge (t)
      (1.5,1) node (rt) {} edge (t)
       (-1.5,-1) node (lb) {} edge (b) edge (lt)
       (1.5,-1) node (rb) {} edge (b)  edge (rt)
       (-1,0) node (l) {} edge (lt)  edge (lb) edge (c)
       (1,0) node (r) {} edge (rt)  edge (rb) edge (c);
     \end{tikzpicture}}$   
  constructed by Roberson \cite[Ex.~8.1]{Geerts21}, one can see that $(\wlk1,\spec)\nrightarrow\ea$.
  Indeed, while $G$ contains two vertices $x$ with $w_2(x,x)=w_3(x,x)=2$, in $H$ we have
  $w_3(x,x)=0$ for every $x$ such that $w_2(x,x)=2$. Therefore, $\wwr0(G)\ne\wwr0(H)$
  and it follows by Part 1 of Corollary \ref{cor:char} that $\ea(G)\ne\ea(H)$.
\end{remark}

We state and prove the three results announced above,
namely Theorems \ref{thm:sep0}, Theorem \ref{thm:sep}, and \ref{thm:sep2},
in the rest of this section.

\subsection{$\wlk{3/2}$ is not stronger than $\strong$}

\begin{theorem}\label{thm:sep0}
  $\strong\npreceq\wlk{3/2}$.
\end{theorem}

\begin{proof}
  The separation $\wlk2\npreceq\wlk{3/2}$ in \cite{RattanS23} is shown by constructing
  a pair of $\wlk{3/2}$-equivalent graphs $G$ and $H$ as follows. Consider two copies
  of $C_6*K_1$, where $*$ is the join of graphs. Denote the vertices of degree 6 by $u_1$
  and $u_2$. Consider also two copies of $(2C_3)*K_1$, denoting their vertices of degree 6 by $v_1$
  and $v_2$. The graph $G$ is obtained by adding four edges forming the cycle $u_1u_2v_1v_2$,
  and the graph $H$ is obtained by adding the cycle $u_1v_1u_2v_2$.
  In \cite{RattanS23} it is observed that $G$ and $H$ are distinguishable by 2-WL.
  We now strengthen this observation to show that $\strong(G)\ne\strong(H)$.

  The vertices $u_1,u_2,v_1,v_2$
  will be referred to as $Q$-vertices. The other vertices are split into two classes, $H$-vertices
  and $T$-vertices, depending on whether they belong to a hexagonal or a triangular part.
  A vertex $x$ is a $Q$-vertex exactly when $w_2(x,x)=\deg x=8$. The $H$- and the $T$-vertices
  are distinguishable by the condition $w_3(x,x)=6$ for a $T$-vertex and $w_3(x,x)=4$ for an $H$-vertex.
  Consider an arbitrary $T$-vertex $x$ in $G$. Note that from $x$ there is at least one 3-walk
  to each of the twelve $T$-vertices $y$. If we consider a $T$-vertex $x$ in $H$, then from $x$ there
  are 3-walks only to six $T$-vertices $y$. This implies that $\wwr1(G)\ne\wwr1(H)$.
  We conclude by Part 3 of Corollary \ref{cor:char} that $G$ and $H$ are not $\strong$-equivalent.
\end{proof}

\subsection{$\wwr2$ is not stronger than $\wlk1$}\label{ss:sep}

\begin{theorem}\label{thm:sep}
  $\wlk1\npreceq\wwr2$.
\end{theorem}

The proof requires a substantial extension of the approach in \cite{ESA23} to separate various
$\wlk1$- and $\wm$-based concepts.

\paragraph{Construction.}
Suppose that we have a graph $A$ with $m$ designated vertices $a_1,\ldots,\allowbreak a_m$ and
a graph $B$ with $m$ designated vertices $b_1,\ldots,b_m$, which will be referred to as
\emph{port vertices}. In each of the graphs, the port vertices are
colored by different colors. Specifically, $a_i$ and $b_i$ are colored by the same color $i$.
The resulting partially colored graphs are denoted by $A'$ and $B'$.
We construct a graph $G(A',B')$ with no colored vertices as follows.
$G(A',B')$ consists of the vertex-disjoint union of $A$ and $B$ and a number
of new vertices of two sorts, namely \emph{connecting} and \emph{pendant vertices}.
For each $i$, there is a connecting vertex $c_i$ adjacent to $a_i$ and $b_i$.
Moreover, for each $i$ there are $i$ pendant vertices $p_{i,1},\ldots,p_{i,i}$ of degree $1$
all adjacent to $c_i$. An example of the construction for $m=3$ is shown in Figure~\ref{fig:G}.

\begin{figure}
\pgfdeclarelayer{background layer}
\pgfsetlayers{background layer,main}

\centering
\begin{tikzpicture}

    \matrix[column sep=3.5mm,row sep=3mm,every node/.style={circle,draw,inner sep=2pt,fill=none}] {
&&&&&&&&&&&&&&&&&&&&&\node (c41) {};&&\\
\node[fill,red] (a4) {};&&&&&&&&\node[fill,red] (b4) {};&&&&&&&&&\node (a4) {};&&&&\node (c4) {};&&&&\node (b4) {};\\
&&&&&&&&&&&&&&&&&&&&\node (c31) {};&&\node (c32) {};\\
\node[fill,blue] (a31) {};&&&&&&&&\node[fill,blue] (b31) {};
&&&&&&&&&\node (a3) {};&&&&\node (c3) {};&&&&\node
(b3) {};\\
&&&&&&&&&&&&&&&&&&&&\node (c11) {};&\node (c12) {};&\node (c13) {};&\\
\node[fill,green] (a1) {};&&&&&&&&\node[fill,green] (b1) {};&&&&&&&&&\node (a1) {};&&&&\node (c1) {};&&&&\node (b1) {};\\[3mm]
\node[draw=none] (A) {$A'$};&&&&&&&&\node[draw=none] (B) {$B'$};&&&&&&&&&&&&&\node[draw=none] (A) {$\!\!\!\!\!\!\!\!\!\!G(A',B')\!\!\!\!\!\!\!\!\!\!$};&&\\
    };

    \draw  (a1) -- (c1) -- (b1)
    (a3) -- (c3) -- (b3)  (a4) -- (c4) -- (b4)
(c1) -- (c11) (c1) -- (c12) (c1) -- (c13)
(c31) -- (c3) -- (c32) (c4) -- (c41) ;

\begin{pgfonlayer}{background layer}
  \draw[fill=gray!30,draw=none] (a31) ellipse (10mm and 15mm);
  \draw[fill=gray!30,draw=none] (b31) ellipse (10mm and 15mm);
  \draw[fill=gray!30,draw=none] (a3) ellipse (10mm and 15mm);
  \draw[fill=gray!30,draw=none] (b3) ellipse (10mm and 15mm);
\end{pgfonlayer}

\end{tikzpicture}
\caption{Construction of $G(A',B')$.}
\label{fig:G}
\end{figure}

\paragraph{Main lemma.}
The crux of the proof is the following lemma, whose proof
is postponed to Subsection~\ref{ss:crux}.
Recall that a \emph{strongly regular graph} with parameters $(n,d,\lambda,\mu)$
is an $n$-vertex $d$-regular graph where every two adjacent vertices
have $\lambda$ common neighbors, and every two non-adjacent vertices
have $\mu$ common neighbors. Extending the notation used in Section \ref{ss:characterization},
for a graph $G$ we set $\kwlout1^r(G)=\Mset{\kwlout1^r(G,x)}_{x\in V(G)}$.

\begin{lemma}\label{lem:main}
  Let $A$ and $B$ be strongly regular graphs with the same parameters (in particular, $A$ and $B$ can be isomorphic).
  Let $A'$ and $B'$ be their partially colored versions such that each color occurs in $A'$, as well as in $B'$, at most once.
  Assume that $\kwlout1^0(A')=\kwlout1^0(B')$, which means that the
  sets of the colors occurring in $A'$ and $B'$ are equal and, therefore,
  we can construct the uncolored graph $G=G(A',B')$. Consider also $H=G(A',A')$
  constructed from two vertex-disjoint copies of~$A'$.
  \begin{enumerate}[\bf 1.]
  \item
    If $\wlkk1(A')\ne\wlkk1(B')$, then $\wlkk1(G)\ne\wlkk1(H)$.
    In words: if color refinement distinguishes $A'$ and $B'$, then it distinguishes also $G$ and~$H$.
  \item
    If $\kwlout1^r(A')=\kwlout1^r(B')$ for some $r\ge1$ (i.e., $r$ rounds of color refinement do not suffice for distinguishing $A'$ and $B'$), then $\wwr{r-1}(G)=\wwr{r-1}(H)$.
  \end{enumerate}
\end{lemma}

\paragraph{The rest of the proof.}
Lemma \ref{lem:main} allows us to separate $\wlk1$ and $\wwr2$ by finding partially colored strongly regular graphs
$A'$ and $B'$ satisfying the assumptions of the lemma for $r=3$.
Let $\mathrm{SRG}(n,d,\lambda,\mu)$ denote the set of strongly regular graphs
with parameters $(n,d,\lambda,\mu)$. Two suitable colorings $A'$ and $B'$
exist for a graph $A=B$ in the set SRG(25,12,5,6) of \emph{Paulus graphs}, namely for
the graph $P_{25.12}$ in Brouwer's collection \cite{BrouwerPaulus}.
This graph was first described in
\cite{Paulus73} and \cite{Rozenfeld73}. It can also be found as the first among the
SRG(25,12,5,6) graphs in McKay's collection of graphs \cite{McKayGraphs} and
the 14th SRG(25,12,5,6) graph in Spence's collection \cite{SpenceSRGraphs}.
The graph can be constructed from the Latin square
\[L=
\left[\begin{matrix}
1 & 2 & 3 & 4 & 5\\
4 & 1 & 5 & 2 & 3\\
5 & 4 & 1 & 3 & 2\\
3 & 5 & 2 & 1 & 4\\
2 & 3 & 4 & 5 & 1
\end{matrix}\right]
\]
by setting $V(A)=\Set{1,2,3,4,5}^2$ and
\[E(A)=\Set{\Set{(i,j),(i',j')}}{i=i' \text{ or } j=j' \text{ or } L_{i,j}=L_{i',j'}}.\]
In this representation, the vertices of $A$ are the cells of the square $L$,
and two cells are adjacent if and only if they are in the same row or in the same column,
or contain the same number in~$L$.

Two appropriate colorings $A'$ and $B'$ of the graph $A=P_{25.12}$ are found by computer search
using the Lua package TCSLibLua~\cite{TCSLibLua}. They use $m=2$ vertex colors.
In the colored version $A'$, the vertex $(1,1)$ is colored red and $(1,3)$ is colored blue.
In $B'$, the vertex $(1,3)$ is red and $(3,4)$ is blue.\footnote{%
It seems that strongly regular graphs with less than 25 vertices
do not admit appropriate colorings even for $r=2$.}

\subsection{Separation of the $\wwr r$ hierarchy from $\wlk2$}

We now prove the following result.

\begin{theorem}\label{thm:sep2}
  $\wlk{3/2}\npreceq\wwr{\bullet}$.
\end{theorem}

Since $\wlk{3/2}\preceq\wlk2$ and $\wwr{r}\preceq\wwr{\bullet}$ for each $r$,
Theorem \ref{thm:sep2} implies that $\wwr r$ is not stronger than $\wlk2$
(for $r=2$ this follows also from Theorem \ref{thm:sep}).
The proof of Theorem \ref{thm:sep2} is based on the following lemma.

\begin{lemma}\label{lem:main2}
  Let $A$ and $B$ be (possibly isomorphic) strongly regular graphs
  with the same parameters.  Let $A'$ and $B'$ be their versions, each
  containing a single individualized (port) vertex. Let $G=G(A',B')$ and
  $H=G(A',A')$.
  \begin{enumerate}[\bf 1.]
  \item
   $\wwr r(G)=\wwr r(H)$ for all $r$ and, therefore, $\wwr{\bullet}(G)=\wwr{\bullet}(H)$.
  \item
    If $\wlkk{3/2}(A')\ne\wlkk{3/2}(B')$, then $\wlkk{3/2}(G)\ne\wlkk{3/2}(H)$.
  \end{enumerate}
\end{lemma}

\begin{proof}
  \textit{1.}
  Note that $\wlk1(A')=\wlk1(B')$.
  Indeed, the distinguishability of $A'$ and $B'$ by 1-WL would imply
  the distinguishability of $A$ and $B$ by 2-WL, contradicting the assumption
  that these strongly regular graphs have the same parameters. Thus,
  $\kwlout1^r(A')=\kwlout1^r(B')$ for all $r\ge1$, and we can apply the second part of Lemma~\ref{lem:main}.

\medskip
  
\textit{2.}  Recall that $G_x$ denotes the graph $G$ with
individualized vertex $x$. Following this notation, the graph $A'_x$
contains the vertex $x$ and the port vertex as individualized
vertices (note that if a graph contains two individualized vertices, then their colors are distinct).
The proof is based on the following.

  \begin{claim}\label{cl:32}
    Let $x\in V(G)$ and $y\in V(H)$. More specifically, suppose that
    $x$ belongs to the $C$-part of $G$, where $C\in\{A,B\}$. Let $C$
    denote the part of $G$ containing $x$, and $D$ denote the part of
    $H$ containing $y$ (thus, $D$ is one of the two copies of $A$ in $H$).
    Similarly to the definition of $A'_x$ above, 
    the graphs $C'_x$ and $D'_y$ are defined by individualizing their port vertices
    (in addition to the individualized vertices $x$ and $y$). Now, if
    $\wlk1(C'_x)\ne\wlk1(D'_y)$ then $\wlk1(G_x)\ne\wlk1(H_y)$.
  \end{claim}

  \begin{subproof}
    Suppose that $A$ and $B$ are $d$-regular.  Denote the
    individualized port vertex of $C'$ by $u$. Thus, $u$ is the vertex
    of degree $d+1$ in the part of $G$ containing $x$. Similarly, let
    $v$ be the individualized vertex of $D'$, that is, the vertex of
    degree $d+1$ in the part of $H$ containing~$y$. Let $\bar{u}$
    and $\bar{v}$ denote the other port vertices in $G$ and $H$,
    respectively.

    Following the notation $\kwlout1^r(G_x,u)$, we let
    $\kwlout1(G_x,u)$ denote the final color assigned to the vertex
    $u$ by the run of 1-WL on the graph $G_x$. The argument is split in two cases.

\Case 1{$\kwlout1(G_x,u)\ne\kwlout1(H_y,v)$.}    
       Note that $A$ and $B$ are
      connected because otherwise $A'$ and $B'$ would be
      isomorphic. This follows from the fact that a disconnected
      strongly regular graph is a disjoint union of complete graphs of
      equal size. Notice that $u$ and $\bar{u}$ are the only vertices
      of degree $d+1$ in $G_x$ and $v, \bar{v}$ the only vertices of
      degree $d+1$ in $H_y$. However, as $u$ is closer to $x$ than
      $\bar{u}$ in the graph $G$, the run of 1-WL on $G_x$ will
      distinguish between $u$ and $\bar{u}$, that is, $\kwlout1(G_x,u)\ne\kwlout1(G_x,\bar u)$.\footnote{%
        If a vertex $w$ is at distance $r$ from the individualized vertex $x$ and
        a vertex $w'$ is at distance larger than $r$ from $x$, then it is easy to see
        that $\kwlout1^r(G_x,w)\ne\kwlout1^r(G_x,w')$.
        }
      Similarly, the run of
      1-WL on $H_y$ will distinguish between $v$ and $\bar{v}$.
      Hence, the port $u$ is the unique vertex of its color in the
      1-WL stable coloring of $G_x$. Likewise, the port $v$ in $H_y$
      is uniquely colored by the stable coloring. Since, by
      assumption, these colors are different, it follows from the connectedness of $G$ and $H$ that
      $\kwlout1(G_x,x)\ne\kwlout1(H_y,y)$. As $x$ and $y$ are
      individualized, this readily implies that $\wlk1(G_x)\ne\wlk1(H_y)$.
      
      \Case 2{$\kwlout1(G_x,u)=\kwlout1(H_y,v)$.}
      The crucial observation is that, in this case, if 1-WL is run on input $G_x$,
      then its action on the part $C$ of $G_x$ is completely similar to the run of
  1-WL on input $C'_x$. The same is true for the action of 1-WL
  on the part $D$ of $H_y$ --- it is completely similar to the run of
  1-WL on input $D'_y$. More formally, it suffices to prove for our purposes that
  for any vertices $w$ in $C'_x$ and $z$ in~$D'_y$,
  \begin{equation}
    \label{eq:GHCD}
\kwlout1^r(C'_x,w)\ne\kwlout1^r(D'_y,z)\text{ implies }\kwlout1^{r+1}(G_x,w)\ne\kwlout1^{r+1}(H_y,z)  
\end{equation}
for any number of rounds $r\ge0$.\footnote{It is also true, on the other hand, that
 $\kwlout1^r(C'_x,w)=\kwlout1^r(D'_y,z)$ implies $\kwlout1^r(G_x,w)=\kwlout1^r(H_y,z)$.   
}
We use the induction on $r$.

In the base case, when $r=0$, assume that $\kwlout1^0(C'_x,w)\ne\kwlout1^0(D'_y,z)$.
This is possible in one of the following two cases.
If $w=x$ and $z\ne y$ or if $w\ne x$ and $z=y$, then the inequality $\kwlout1^{1}(G_x,w)\ne\kwlout1^{1}(H_y,z)$
is true because $x$ and $y$ are individualized in $G_x$ and $H_y$ respectively.
The other possibility is that $w=u$ and $z\ne v$ or that $w\ne u$ and $z=v$.
Then $\kwlout1^{1}(G_x,w)\ne\kwlout1^{1}(H_y,z)$ because $u$ and $v$ are distinguishable
by having degree~$d+1$.

Now, let $r\ge1$ and assume that \refeq{GHCD} is true for the preceding value of $r$.
Assume that
\begin{equation}
  \label{eq:CneD}
\kwlout1^r(C'_x,w)\ne\kwlout1^r(D'_y,z).  
\end{equation}
If $\kwlout1^{r-1}(C'_x,w)\ne\kwlout1^{r-1}(D'_y,z)$, then the induction assumption yields
the inequality $\kwlout1^{r}(G_x,w)\ne\kwlout1^{r}(H_y,z)$, which implies that
$\kwlout1^{r+1}(G_x,w)\ne\kwlout1^{r+1}(H_y,z)$. Next, suppose that $\kwlout1^{r-1}(C'_x,w)=\kwlout1^{r-1}(D'_y,z)$.
It follows that either $w=u$ and $z=v$ or that $w\ne u$ and $z\ne v$. The latter subcase is somewhat
simpler, and we consider it first.

\subcase{(i)}{$w\ne u$ and $z\ne v$.}
For a graph $X$ and a vertex $a\in V(X)$, let $N_X(a)$ denote the neighborhood of $a$ in~$X$.
Inequality \refeq{CneD} implies that
\begin{equation}
  \label{eq:MsCneMsD}
\Mset{\kwlout1^{r-1}(C'_x,g)}_{g\in N_C(w)}\ne\Mset{\kwlout1^{r-1}(D'_y,h)}_{h\in N_D(z)}.
\end{equation}
By the induction assumption, this gives us the inequality
\begin{equation}
  \label{eq:GgHh}
\Mset{\kwlout1^{r}(G_x,g)}_{g\in N_G(w)}\ne\Mset{\kwlout1^{r}(H_y,h)}_{h\in N_H(z)},  
\end{equation}
from which we derive the desired inequality $\kwlout1^{r+1}(G_x,w)\ne\kwlout1^{r+1}(H_y,z)$.

\subcase{(ii)}{$w=u$ and $z=v$.}
Denote the connecting vertex in $G$ by $c$ and the connecting vertex in $H$ by $c'$.
By construction, both $c$ and $c'$ have degree 3. This degree is unique in $G$ and in $H$
because $d>3$. Indeed, the only connected strongly regular graphs of degree 2 are
the cycles $C_3$ and $C_5$, and the only connected strongly regular graphs of degree 3 are
the complete graph $K_4$, the bipartite complete graph $K_{3,3}$, and the Petersen graph;
see, e.g., \cite{HoltonL92}. All of these graphs are vertex-transitive,
where a graph is \emph{vertex-transitive} if every its vertex can be taken to any other vertex by
some automorphism. However, in such a case the graphs $A'$ and $B'$ would be isomorphic,
contradicting the assumption $\wlkk{3/2}(A')\ne\wlkk{3/2}(B')$ of the lemma.

Note that $N_G(u)=N_C(u)\cup\{c\}$ and $N_H(v)=N_D(v)\cup\{c'\}$.
In the subcase under consideration, Inequality \refeq{MsCneMsD} alone does not give us
Inequality \refeq{GgHh}, implying instead its weaker version
\begin{equation}
  \label{eq:GgHhuv}
\Mset{\kwlout1^{r}(G_x,g)}_{g\in N_G(u)\setminus\{c\}}\ne\Mset{\kwlout1^{r}(H_y,h)}_{h\in N_H(v)\setminus\{c'\}}.
\end{equation}
Nevertheless, we have the equality
\begin{equation}\label{eq:ccc}
  \kwlout1^{r}(G_x,c)=\kwlout1^{r}(H_y,c'),
\end{equation}
which follows from the assumption that $\kwlout1^{r+1}(G_x,u)=\kwlout1^{r+1}(H_y,v)$
because $c$ is the only vertex of degree 3 in the neighborhood of $u$ and
$c'$ is the only vertex of degree 3 in the neighborhood of $v$.
Inequality \refeq{GgHhuv} and Equality \refeq{ccc} together imply Inequality \refeq{GgHh},
and the proof is finished like in the Subcase~(i).

\medskip

Thus, \refeq{GHCD} is proved for all $r$. It follows that for any vertices $w$ in $C'_x$
      and $z$ in $D'_y$ the inequality $\kwlout1(C'_x,w)\ne \kwlout1(D'_y,z)$
      implies the inequality $\kwlout1(G_x,w)\ne \kwlout1(H_y,z)$.
Therefore, the inequality $\wlk1(C'_x)\ne \wlk1(D'_y)$ implies the inequality
$\wlk1(G_x)\ne \wlk1(H_y)$.
Indeed, if the multisets of colors in the stable coloring of $G_x$ and $H_y$ are
equal, then their restrictions to $C$ and $D$ are equal as well, and
the multisets of colors in $C'_x$ and $D'_y$ have to be equal too.
  \end{subproof}

Recall that $\wlk{3/2}(A')=\Mset{\wlk1(A'_x)}_{x\in V(A)}$ and
$\wlk{3/2}(B')=\Mset{\wlk1(B'_x)}_{x\in V(B)}$.  Thus, the assumption
of Part 2 of the lemma implies that there exists a vertex $z\in V(A)$ such
that $\wlk1(A'_z)$ occurs in $\wlk{3/2}(A')$ more frequently than in
$\wlk{3/2}(B')$. Specifically, let $z_1,\ldots,z_s$ be all vertices in
$A$ such that $\wlk1(A'_{z_i})=\wlk1(A'_z)$.  Let $\tilde z_i$ denote
the clone of $z_i$ in the other $A$-part of $H$. Consider the
corresponding elements of $\wlk{3/2}(H)$, namely
$\wlk1(H_{z_1}),\ldots,\wlk1(H_{z_s}),\wlk1(H_{\tilde
  z_1}),\ldots,\wlk1(H_{\tilde z_s})$.  By Claim \ref{cl:32}, whatever
these elements are, they cannot occur (with the same frequency)
in~$\wlk{3/2}(G)$.
\end{proof}

\begin{proof}[Proof of Theorem \ref{thm:sep2}]
  We apply Lemma \ref{lem:main2}, where $A$ is the Shrikhande graph
  and $B$ is the $4\times4$ rook's graph, both strongly regular graphs
  with parameters $(16,6,2,2)$. Since both graphs are
  vertex-transitive, $A'$ and $B'$ are uniquely defined.
  We have $\wwr{\bullet}(G)=\wwr{\bullet}(H)$ by Part 1 of Lemma \ref{lem:main2}.
  To ensure that $\wlk{3/2}(G)\ne\wlk{3/2}(H)$, by Part 2 of this lemma we have to show that $\wlk{3/2}(A')\ne\wlk{3/2}(B')$.
  Let $a$ and $b$ denote the vertices individualized in $A'$ and $B'$ respectively.

  Recall that a graph is \emph{arc-transitive} if every ordered pair of its adjacent vertices can be taken
  to any other ordered pair of adjacent vertices by some automorphism of the graph.
  Since the $4\times4$ rook's graph and its complement are both arc-transitive, $\wlk1(B'_x)$ takes on exactly three values
  depending on whether $x=b$ and whether $x$ is adjacent with $b$. The Shrikhande graph is also arc-transitive,
  but this is no longer true for its complement. Specifically, when we consider
  the action of the automorphism group of the Shrikhande graph on $V(A)^2$, then the set of all pairs of
  non-adjacent unequal vertices $(x,y)$ is split into two orbits depending on whether the two
  common neighbors of $x$ and $y$ are adjacent or not. For $y=a$, let $x_1$ be as in the first case
  and $x_2$ be as in the second case. Clearly, $\wlk1(A'_{x_1})\ne\wlk1(A'_{x_2})$.
  It follows that $\wlk1(A'_x)$ takes on exactly four different values and, therefore,
  the multisets  $\wlk{3/2}(A')$ and $\wlk{3/2}(B')$ are different from each other.
\end{proof}

\begin{open}\label{probl:main}
  We leave open the question whether the hierarchy
  \begin{equation}
    \label{eq:wwr-hierarchy}
\wwr1\preceq\wwr2\preceq\wwr3\preceq\wwr4\preceq\cdots\preceq\wwr{\bullet}
  \end{equation}
is strict or at least does not collapse to some level.
While we know that each $\wwr r$ is strictly weaker than $\wlk2$,
it remains open whether $\wwr r$ can be stronger than $\wlk1$
for some large $r$. A negative answer will follow from Lemma \ref{lem:main}
if there is an infinite sequence of partially colored strongly regular graphs $A'_r$ and $B'_r$
for $r=1,2,3\ldots$, where the underlying graphs are equal or have the same parameters, such that
$A'_r$ and $B'_r$ are distinguished by 1-WL, but requiring at least $r$ refinement rounds.
\end{open}

Suppose that strongly regular graphs $A$ and $B$ have the same parameters and their partially colored
versions $A'$ and $B'$ are distinguished by 1-WL exactly in the $(r+1)$-th round.
By Lemma \ref{lem:main}, $G(A',B')$ and $G(A',A')$ are $\wwr{r-1}$-equivalent and, therefore,
this pair of graphs is a good candidate for separation of $\wwr{r-1}$ from $\wwr r$.
This approach works indeed pretty well.

\begin{theorem}\label{thm:lower-levels}
  The hierarchy \refeq{wwr-hierarchy} is strict up to the 4-th level, that is,
  the first three relations in \refeq{wwr-hierarchy} are strict.
\end{theorem}

\begin{proof}
For each $r=1,2,3$, we provide a pair of graphs $G_r$ and $H_r$ such that
\begin{equation}
  \label{eq:wwrr}
\wwr{r}(G_r)=\wwr{r}(H_r)  
\end{equation}
while
\begin{equation}
  \label{eq:wwrr1}
\wwr{r+1}(G_r)\ne\wwr{r+1}(H_r).
\end{equation}
For $r=1$ and $r=2$, these graphs are obtained by using our construction,
that is, $G_r=G(A'_r,B'_r)$ and $H_r=G(A'_r,A'_r)$ for suitable partially
colored strongly regular graphs $A'_r$ and $B'_r$. For both $r=1$ and $r=2$,
$A'_r$ and $B'_r$ are colored versions of the Paulus graph $P_{25.12}$
that was already used in the proof of Theorem \ref{thm:sep}.
$A'_2$ and $B'_2$ are exactly as in the proof of Theorem \ref{thm:sep}.
$A'_1$ and $B'_1$ are obtained by using 3 colors, specifically,
by setting $a_i=(1,i)$ and $b_i=(i,1)$ for $i=1,2,3$ (where $a_i$ and $b_i$ are
the vertices colored by the color $i$ in $A'_1$ and $B'_1$ respectively).
Here we use the representation of $P_{25.12}$ as a Latin square graph
given in Subsection \ref{ss:sep}.
While we already know that $A'_2$ and $B'_2$ satisfy the assumptions of
Lemma \ref{lem:main} for $r=3$, a direct computation shows that
$A'_1$ and $B'_1$ satisfy the assumptions of Lemma \ref{lem:main} for $r=2$.
Thus, Equality \refeq{wwrr} follows from Lemma \ref{lem:main} in both cases
of $r=1$ and $r=2$. Inequality \refeq{wwrr1} is in both cases verified directly
by computation.\footnote{The computations are done by using the Lua package TCSLibLua \cite{TCSLibLua}.
  More specifically, $w(x,y)$ is computed using the method \url{graphs::walkMatrix} with $\mathrm{e}_y$ as
  \url{filter} for each vertex $y$; the graph invariant $\wwr{r}$ is computed using the \url{genref} module.}

For $r=3$, we use an extension of the construction of $G(A',B')$, allowing
color classes of size larger than one. Specifically, for each color $i$ we now allow
possibly more than one vertex $a_i^1,a_i^2,\ldots$ colored by $i$ in $A'$
and the same number of vertices $b_i^1,b_i^2,\ldots$ colored by $i$ in $B'$.
The connecting vertex $c_i$ is adjacent to $a_i^j$ and $b_i^j$ for all $j$.
Since we do not prove an analog of Lemma \ref{lem:main} for the generalized construction,
we will need to verify Equality \refeq{wwrr} and Inequality \refeq{wwrr1} by computation.
$A'_3$ and $B'_3$ are colored versions of another Paulus graph $P_{25.02}$.
This graph and its complement are the only two graphs in SRG(25,12,5,6)
with trivial automorphism group. We work with a particular labeled version of $P_{25.02}$,
namely with a graph
  \verb|X}rU\adeSetTjKWNJEYNR]PLjPBgUGVTkK^YKbipMcxbk`{DlXF|
  in graph6 notation. Assuming vertex names $1,\dots,25$
  (the graph6 notation defines an order on the vertices),
  we color the vertices $2$ and $6$ in $A'$ and $3$ and $22$ in $B'$
  all with the same color; the other vertices are uncolored. The graphs
  $G_3=G(A'_3,B'_3)$ and $H_3=G(A'_3,A'_3)$ are obtained by using the generalized construction.
  Equality \refeq{wwrr} and Inequality \refeq{wwrr1} are verified with help of the Lua package
  TCSLibLua~\cite{TCSLibLua}.
\end{proof}

In fact, the same approach can also be used to strengthen Theorem \ref{thm:sep}.

\begin{theorem}\label{thm:sep-explicit-3}
  $\wlk1\npreceq\wwr3$.
\end{theorem}
\begin{proof}
  We have to exhibit two graphs $G$ and $H$ such that $\wwr{3}(G)=\wwr{3}(H)$
  while $\wlk{1}(G)\ne\wlk{1}(H)$. For this purpose, we take $G=G_3$ and $H=H_3$
  as in the preceding proof and verify that both conditions are true by using the Lua package
  TCSLibLua~\cite{TCSLibLua}.
\end{proof}

\subsection{Proof of Lemma \ref{lem:main}}\label{ss:crux}

1.
Though the colors of port vertices in $A'$ and $B'$ are removed after these graphs are
combined into $G$ and $H$, these colors are simulated by CR due to the adjacency to
the respective connection vertex $c_i$. Recall that $c_i$ bears the information about
the color $i$ as it is adjacent to exactly $i$ pendant vertices. Since the vertex-colored graphs
$A'$ and $B'$ are distinguishable by CR, every vertex in the $B'$-part of $G$ will
eventually receive a color not occurring in~$H$.

\medskip

2.
The proof of Part 2 is based on Claims \ref{cl:srg}, \ref{cl:A}, and \ref{cl:B} below.
We begin with an overall outline of the proof, explaining the role of each of these claims.
We need to prove the implication
\begin{equation}
  \label{eq:newimpl}
  \kwlout1^r(A')=\kwlout1^r(B')\implies \wwr{r-1}(G)=\wwr{r-1}(H).
\end{equation}
Let us focus on the case of $r=1$. That is, under the assumption $\kwlout1^1(A')=\kwlout1^1(B')$
we have to show that $\wwr{0}(G)=\wwr{0}(H)$.

Let $x$ be a vertex in $G$ or $H$, and suppose that $x$ is neither connecting nor pendant.
Suppose that $x$ belongs to the part $P$ of $G$ or $H$. Thus, $P$ is an isomorphic copy of $A$ or $B$.
Recall that the vertex $x$ contributes the vector $w_*(x,x)=\of{ w_0(x,x),w_1(x,x),\ldots,w_{n-1}(x,x) }$
in the multiset $\wwr{0}(G)$ or $\wwr{0}(H)$, where $w_k(x,x)$ is the number of closed $k$-walks
from $x$ to itself. Consider first the set of $k$-walks from $x$ to itself that do not leave $P$.
The number of such walks is the same for any choice of $x$ because the number of $k$-walks from
$x$ to another vertex $z$ in a strongly regular graph with given parameters is determined by the adjacency
and equality relations between $x$ and $z$ (Claim \ref{cl:srg}).

If we, furthermore, want to count $k$-walks emanating from $x$, leaving $P$, and eventually coming back
to $x$, then we have to take into account that such a walk leaves $P$ via one of the port vertices $z$.
An essential information is whether $x=z$ and, if not, then whether $x$ and $z$ are adjacent or not.
Claims \ref{cl:srg} and \ref{cl:A} allow us to show that $w_k(x,x)$ is determined by the CR-color
$\kwlout1^1(P',x)$ (note that each port vertex $z$ is individualized in $P'$ and, therefore,
the color $\kwlout1^1(P',x)$ is nothing else as the list of the adjacency/equality relations between $x$ and all $z$).
Since the equality $\kwlout1^1(A')=\kwlout1^1(B')$ allows us to establish a matching between
the vertices of $G$ and $H$ preserving the color $\kwlout1^1(P',x)$, from here we easily
derive the equality $\wwr{0}(G)=\wwr{0}(H)$.

Finally, Claim \ref{cl:B} allows us to prove the implication \refeq{newimpl} for all $r\ge1$
by induction on $r$. We now proceed to a detailed proof.

Recall that we write $A$ and $I$ to denote the adjacency and the equality relations.
We use the following well-known property of strongly regular graphs.

\setcounter{claim}{0}

\begin{claim}\label{cl:srg}
  Let  $w_k^X(x,y)$ be the number of $k$-walks from a vertex $x$ to a vertex $y$
  in a strongly regular graph $X$. The walk count $w_k^X(x,y)$ is a function of $k$,
the predicate values $I(x,y)$ and $A(x,y)$, and the parameters of $X$
(and is otherwise independent of $x$ and~$y$).
\end{claim}

In Claims \ref{cl:A} and \ref{cl:B} below we establish some crucial properties of our construction of $G(A',B')$.
For a vertex $z$ of a colored graph $F$, let $C_r^F(z)=\kwlout1^r(F,z)$ denote
the color of $z$ after performing $r$ rounds of CR on input $F$.
For every vertex $x$ of $G(A',B')$ we define its type $\tau(x)$ as follows.
If $x$ is in the $A$-part, then $\tau(x)=C_1^{A'}(x)$.
Similarly, $\tau(x)=C_1^{B'}(x)$ if $x$ belongs to the $B$-part.
In order to define the types of connecting and pendant vertices, we use special symbols
$\mathfrak{c}$ and $\mathfrak{p}$ and
set $\tau(c_i)=(\mathfrak{c},\tau(a_i))$ and $\tau(p_{i,h})=(\mathfrak{p},\tau(a_i))$ for all $h\le i$.
We define a binary relation $D$ on the vertex set of $G(A',B')$
by setting $D(x,y)=1$ exactly when $x$ and $y$ are both in the $A$-part or both in the $B$-part.
Finally, we define a coloring $\eta$ of vertex pairs in $G(A',B')$ by
$$
\eta(x,y)=\of{ \tau(x), D(x,y), I(x,y), A(x,y), \tau(y) }.
$$
We also set $\tau(A')=(\tau(a_1),\ldots,\tau(a_m))$ and denote the tuple of the four parameters
of the strongly regular graph $A$ (hence also of $B$) by~$\sigma(A)$.

\begin{claim}\label{cl:A}
  Suppose that $\tau(a_i)=\tau(b_i)$ for all $i\le m$ (or, equivalently, that the mapping $a_i\mapsto b_i$
  is a partial isomorphism between the colored graphs $A'$ and $B'$).
For vertices $x$ and $y$ of $G(A',B')$, the sequence $w_*(x,y)$ is solely determined by
$\eta(x,y)$, $\tau(A')$, and~$\sigma(A)$.
\end{claim}

\begin{subproof}
By induction on $k$, we prove that $w_k(x,y)$ is determined by $\eta(x,y)$, $\tau(A')$, and $\sigma(A)$.
The base case of $k=0$ is trivial. For the induction step, suppose that $k\ge1$.

\Case1{$D(x,y)=1$.}
Assume that $x$ is in the $A$-part; the case of $x\in V(B)$ is similar.
\subcase{1-a}{$x=y$.}
We split the set of closed $k$-walks from $x$ to itself in $G(A',B')$ into two classes:
those remaining all the time in $A$ and those leaving $A$ at some point.
For the first type of walks, we use Claim \ref{cl:srg}, which implies that
their number, i.e., $w_k^A(x,x)$, depends only on $k$ and $\sigma(A)$ and is independent of~$x$.

Let $w_k^\triangledown(x,x)$ denote the number of closed walks from $x$ leaving $A$ at some point.
Since a walk leaves $A$ only via one of the $m$ port vertices $a_i$
and visits the connecting vertex $c_i$ immediately after this, we have
\begin{equation}
  \label{eq:wkxx-raus}  
w_k^\triangledown(x,x)=\sum_{i=1}^m\sum_{s=0}^{k-2}\sum_{t=1}^{k-1-s} w_s^A(x,a_i)w_t(c_i,x).
\end{equation}
Claim \ref{cl:srg} ensures, for each $s$, that $w_s^A(x,a_i)$ is determined by
$\sigma(A)$ and the pair of boolean values $I(x,a_i)$ and
$A(x,a_i)$. Since the last pair is determined by $C_1^{A'}(x)$,
the whole
sequence $w_s^A(x,a_1),\ldots,w_s^A(x,a_m)$ is determined by
$\tau(x)=C_1^{A'}(x)$.  The sequence $w_t(c_1,x),\ldots,w_t(c_m,x)$ is
determined for each $t<k$ due to the knowledge of $\tau(x)$ and $\tau(A')$ by the
induction hypothesis (note that $\eta(c_i,x)$ for all $i$ are
obtainable from $\tau(x)$ and~$\tau(A')$).

We conclude that the count $w_k(x,x)=w_k^A(x,y)+w_k^\triangledown(x,x)$ is determined by
$\tau(A')$, $\sigma(A)$, and $\tau(x)$, where $\tau(x)$ is a part of~$\eta(x,x)$.

We define Subcase 1-b by the condition $A(x,y)=1$ and Subcase 1-c by $A(x,y)=I(x,y)=0$.
The analysis of both these subcases is quite similar to Subcase~1-a,
by defining the counter $w_k^\triangledown(x,y)$ similarly to $w_k^\triangledown(x,x)$
and using the straightforward generalization of Equality~\refeq{wkxx-raus}
(cf.~also Equality~\refeq{wkxy-raus} below).

\Case2{$D(x,y)=0$.}
\subcase{2-a}{$x\in V(A)$ and $y\in V(B)$.} (The case $x\in V(B), y\in V(A)$ is similar.)
Since a walk from $x$ to $y$ must leave $A$ somewhere, we have
\begin{equation}
  \label{eq:wkxy-raus}
w_k(x,y)=\sum_{i=1}^m\sum_{s=0}^{k-2}\sum_{t=1}^{k-1-s} w_s^A(x,a_i)w_t(c_i,y),
\end{equation}
which enables the same analysis as in Subcase~1-a.

\subcase{2-b}{$x\in V(A)$ while $y=c_j$ or $y=p_{j,h}$.} (If $x\in V(B)$, the analysis is virtually the same.)
By the same reasons as in Subcase 2-a, we have
$$
w_k(x,y)=\sum_{i=1}^m\sum_{s=0}^{k-1}\sum_{t=0}^{k-1-s} w_s^A(x,a_i)w_t(c_i,y),
$$
which makes the above argument applicable.

\subcase{2-c}{$x=c_j$ and $y=c_{j'}$.}
In this case, the induction step is based on the equality
$$
w_k(c_j,c_{j'})=w_{k-1}(a_j,c_{j'})+w_{k-1}(b_j,c_{j'})+\sum_{h=1}^j w_{k-1}(p_{j,h},c_{j'}).
$$

\subcase{2-d}{$x=p_{j,h}$ and $y=c_{j'}$.}
For the induction step, note that $w_k(p_{j,h},c_{j'})=w_{k-1}(c_j,c_{j'})$.

\subcase{2-e}{$x=p_{j,h}$ and $y=p_{j',h'}$.}
Similarly to Subcase 2-d, we use the equality
$w_k(p_{j,h},p_{j',h'})=w_{k-1}(c_j,p_{j',h'})$.
\end{subproof}

For ease of notation, set $\omega(x,y)=w_*(x,y)$. For $r\ge0$, let
$\omega_r$ be the vertex coloring as defined in~\refeq{chi}.

\begin{claim}\label{cl:B}
  Suppose that $\tau(a_i)=\tau(b_i)$ for all $i\le m$.
  For a vertex $x$ of $G(A',B')$, the color $\omega_r(x)$ depends on $\tau(A')$ and $\sigma(A)$.
  Once these parameters are fixed, $\omega_r(x)$ is determined by the pair of multisets $M_{A,r}=\Mset{C^{A'}_{r}(y)}_{y\in V(A)}$
  and $M_{B,r}=\Mset{C^{B'}_{r}(y)}_{y\in V(B)}$ and
  \begin{itemize}
  \item
    by the color $C^{A'}_{r+1}(x)$ if $x$ is in the $A$-part,
  \item
     by the color $C^{B'}_{r+1}(x)$ if $x$ is in the $B$-part,
  \item
    by the type $\tau(x)$ if $x=c_i$ or $x=p_{i,h}$.
  \end{itemize}
\end{claim}

\begin{subproof}
  We proceed by induction on $r$. In the base case we have $r=0$.
  If $x\in V(A)$ or $x\in V(B)$, then either $C^{A'}_{1}(x)$ or, respectively, $C^{B'}_{1}(x)$
  determines $\omega_0(x)=w_*(x,x)$ by Claim \ref{cl:A}.
  If $x=c_i$ or $x=p_{i,h}$, then $\tau(x)$ determines $\omega_0(x)$,
  again by Claim~\ref{cl:A}.

  For the induction step, suppose that $r\ge1$ and recall that
  $$
\omega_r(x)=\of{\omega_{r-1}(x),\Mset{(\omega(x,y),\omega_{r-1}(y))}_{y}}.
  $$

\Case1{$x\in V(A)$.} (The case of $x\in V(B)$ is similar.)
By the induction hypothesis, $\omega_{r-1}(x)$ is determined by $M_{A,r-1}$, $M_{B,r-1}$, and $C^{A'}_{r}(x)$,
which are obviously obtainable from $M_{A,r}$, $M_{B,r}$, and $C^{A'}_{r+1}(x)$.
It remains to show that
\begin{equation}
  \label{eq:msety}
\Mset{(\omega(x,y),\omega_{r-1}(y))}_{y}
\end{equation}
is determined as well. To this end, we split this multiset into several parts.

For $y=x$, the corresponding element $(\omega_0(x),\omega_{r-1}(x))$ of the multiset
is determined as we have already shown.

Consider the restriction of \refeq{msety} to $y\in V(A)\setminus\{x\}$.
Along with $M_{A,r}$, the color $C^{A'}_{r+1}(x)$ determines the multiset
$\Mset{(A(x,y),C^{A'}_r(y))}_{y\in V(A)\setminus\{x\}}$. As easily seen from the definition
of $\eta(x,y)$, a pair $(A(x,y),C^{A'}_r(y))$ yields the pair $(\eta(x,y),C^{A'}_r(y))$.
By Claim \ref{cl:A}, $\eta(x,y)$ determines $\omega(x,y)$. By the induction hypothesis,
$C^{A'}_r(y)$ determines $\omega_{r-1}(y)$ (by taking into account that $M_{A,r-1}$ and $M_{B,r-1}$
are obtainable from $M_{A,r}$ and $M_{B,r}$). Thus, we do obtain \refeq{msety} over $y\in V(A)\setminus\{x\}$.

Next, consider the restriction of \refeq{msety} to $y\in V(B)$.
The multiset $M_{B,r}$ readily provides us with the multiset $\Mset{(\eta(x,y),C^{B'}_r(y))}_{y\in V(B)}$.
Similarly to the above, this gives us \refeq{msety} over $y\in V(B)$ by Claim \ref{cl:A} and the
induction hypothesis.

Now, we have to reconstruct the restriction of \refeq{msety} to $y$ ranging over all connecting vertices,
that is, the multiset $\Mset{(\omega(x,c_i),\omega_{r-1}(c_i))}_{i=1}^m$.
The type vector $\tau(A')$ provides us with the sequence of $\of{\eta(x,c_i),\tau(c_i)}$
for $i=1,\ldots,m$. For each $i$, $\eta(x,c_i)$ determines $\omega(x,c_i)$ by Claim \ref{cl:A}.
The induction assumption implies that $\omega_{r-1}(c_i)$ is determined by $\tau(c_i)$
along with $M_{A,r-1}$ and $M_{B,r-1}$ (obtainable from $M_{A,r}$ and $M_{B,r}$).

The restriction of \refeq{msety} to pendant vertices is reconstructible in the same way.
Summing up, we see that the multiset \refeq{msety} for $y$ ranging over the vertex set of $G(A',B')$
is determined.

\Case2{$x=c_i$ for some $i\le m$.}
Similarly to the preceding case, $\omega_{r-1}(c_i)$ is determined by the induction hypothesis
and we only have to consider the multiset $\Mset{(\omega(c_i,y),\omega_{r-1}(y))}_{y}$.
We again split this multiset into several parts and reconstruct each of them separately.

The restriction to $y\in V(A)\cup V(B)$ is reconstructible because $M_{A,r}$ and $M_{B,r}$
give us the multisets $\Mset{(\eta(c_i,y),C^{A'}_r(y))}_{y\in V(A)}$
and $\Mset{(\eta(c_i,y),C^{B'}_r(y))}_{y\in V(B)}$, which allows using Claim \ref{cl:A}
and the induction hypothesis.

Now, we have to reconstruct the restriction to the connection vertices, i.e.,
$\Mset{(\omega(c_i,c_j),\omega_{r-1}(c_j))}_{j=1}^m$.
The type $\tau(c_i)$ and the type vector $\tau(A')$ give us the sequence $\eta(c_i,c_j)$
for $j=1,\ldots,m$, which determines the sequence of $\omega(c_i,c_j)$ for $j=1,\ldots,m$ by Claim \ref{cl:A}.
By the induction hypothesis, $\tau(A')$ along with the multisets $M_{A,r-1}$ and $M_{B,r-1}$
determines the sequence of $\omega_{r-1}(c_j)$ for $j=1,\ldots,m$.
The restriction to the pendant vertices is reconstructible by a similar argument.

\Case3{$x=p_{i,h}$ for some $i\le m$ and $h\le i$.}
This case is similar to Case~2.
\end{subproof}

To complete the proof of Part 2 of Lemma \ref{lem:main},
consider $G=G(A',B')$ and $H=G(A',A')$ for partially colored strongly regular
graphs $A'$ and $B'$ satisfying the conditions of the lemma. In particular,
\begin{equation}
  \label{eq:aabb}
\kwlout1^r(A')=\kwlout1^r(B'),
\end{equation}
where $r\ge1$.
This equality implies that $\kwlout1^1(A')=\kwlout1^1(B')$ and, hence,
$\tau(a_i)=\tau(b_i)$ for all $i\le m$, which makes Claim \ref{cl:B} applicable.
Another consequence of \refeq{aabb} is that the multiset $M_{A,r}$, as well as $M_{B,r}$,
is the same for $G$ and $H$.
In particular, there is a bijection $f\function{V(A)}{V(B)}$ such that
\begin{equation}
  \label{eq:f}
C^{A'}_r(x)=C^{B'}_r(f(x)).
\end{equation}
Define a bijection $F\function{V(H)}{V(G)}$
as the identity on the first $A$-part of $H$ (which we identify with the $A$-part of $G$)
and as the mapping of the second $A$-part of $H$ to the $B$-part of $G$ according to the bijection $f$.
Furthermore, $F$ matches the connecting vertices for the same colors as well as
their pendant vertices. By taking into account \refeq{f}
and the fact that $G$ and $H$ share the same multisets $M_{A,r-1}$ and $M_{B,r-1}$, Claim \ref{cl:B}
readily implies that $\omega_{r-1}(x)=\omega_{r-1}(F(x))$ for all vertices $x$ of $H$.
This gives us the desired equality $\omega^{(r-1)}(G)=\omega^{(r-1)}(H)$.


\end{document}